\documentclass[envcountsame]{llncs}

\usepackage{ifthen}
\newcommand{\techRep}{true} 
\newcommand{\iftechrep}{\ifthenelse{\equal{\techRep}{true}}}

\usepackage{complexity}
\usepackage{multirow}
\usepackage{algorithmic}
\usepackage{algorithm}
\usepackage{url}
\usepackage{btran}
\usepackage{tikz}

\usepackage{mathtools}
\usepackage{amsfonts}
\usepackage{xspace}
\usepackage{pstricks}

\newcommand{\F}{\mathcal{F}}
\newcommand{\Min}{\mathit{Min}}
\newcommand{\Mt}{\DTMC_\mathit{type}}
\newcommand{\Mp}{\DTMC_\mathit{proc}}
\renewcommand{\S}{\mathcal{S}}
\renewcommand{\P}{\mathcal{P}}
\newcommand{\Pt}{\P_\mathit{type}}
\newcommand{\Pp}{\P_\mathit{proc}}
\newcommand{\Prob}{\mathit{Prob}}
\newcommand{\Q}{\mathbb{Q}}
\newcommand{\Run}{\mathit{Run}}
\newcommand{\deltat}{\delta_\mathit{type}}
\newcommand{\deltap}{\delta_\mathit{proc}}
\newcommand{\tran}[1]{\xrightarrow{#1}}
\newcommand{\U}{\mathsf{U}}
\newcommand{\Wt}[1]{|#1|_\mathit{type}}
\newcommand{\Wp}[1]{|#1|_\mathit{proc}}
\newcommand{\dist}{\mathit{dist}}
\renewcommand{\D}{\mathcal{D}}
\newcommand{\En}{\mathit{En}}
\newcommand{\last}{\mathit{last}}

\newcommand\uc{\uparrow}
\newcommand\eps{\varepsilon}

\newcommand\controls{\mathcal{Q}}

\newcommand\cmrules{\Delta} 

\newcommand{\up}{\mathord{\uparrow}}

\newenvironment{qtheorem}[1]{%
{\par\medskip\noindent\bf Theorem #1.}
\begin{itshape}%
}{%
\end{itshape}%
}
\newenvironment{qlemma}[1]{%
{\par\medskip\noindent\bf Lemma #1.}%
\begin{itshape}%
}{%
\end{itshape}%
}

\newenvironment{qproposition}[1]{%
{\par\medskip\noindent\bf Proposition #1.}
\begin{itshape}%
}{%
\end{itshape}%
}

\newcommand{\defn}[1]{\textit{#1}}

\newcommand{\Z}{\ensuremath{\mathbb{Z}}}
\newcommand{\N}{\ensuremath{\mathbb{N}}}

\newcommand{\OMIT}[1]{}

\newcommand{\DTMC}{\ensuremath{\mathcal{M}}}

\newcommand{\CM}{\ensuremath{{M}}}

\newcommand{\Sign}{\ensuremath{\text{sign}}}

\title{Analysis of Probabilistic Basic Parallel Processes}
\author{R{\'e}mi Bonnet\inst{1} \and Stefan Kiefer\inst{1}\thanks{Stefan Kiefer is supported by a Royal Society University Research Fellowship.} \and Anthony W.
    Lin\inst{1,2}}
\institute{University of Oxford, UK \and Academia Sinica, Taiwan}

\begin{document}

\maketitle

\begin{abstract}
    \OMIT{
Basic Parallel Processes (BPPs) are a well-known subclass of Petri Nets
where inter-process communications are restricted. They are the simplest
common model of concurrent programs that allows unbounded spawning of processes.
In this paper, we propose a natural probabilistic extension of BPPs, over which
we study several fundamental qualitative problems (i.e. checking if the
probabilities of certain events are 0 or 1).
Introducing probability in BPPs is motivated by the usefulness of randomisation
when designing
distributed protocols (e.g. for breaking symmetry).
%
Probabilistic BPPs come in two flavours: (i) Markov
Chains, and (ii) Markov Decision Processes (MDPs). For the Markov Chain version,
we show decidability of qualitative reachability with respect to upward closed
target sets, though we show a nonprimitive recursive lower bound. When
semilinear target
sets are used, we have undecidability. In the case of MDPs, we consider
two cases: angelic or demonic schedulers. In both cases, we show decidability
with respect to upward closed target sets assuming bounded fairness (without
this assumption, we show that qualitative reachability is sensitive to
perturbation of the probabilities of individual transitions).  Lastly, we
show natural restrictions of target sets that yield good complexity.
%
%
}

Basic Parallel Processes (BPPs) are a well-known subclass of Petri Nets.
They are the simplest common model of concurrent programs that allows unbounded
spawning of processes.
In the probabilistic version of BPPs,
every process generates other processes according to a probability distribution.
We study the decidability and complexity of fundamental qualitative problems
over probabilistic BPPs --- in particular reachability with probability 1 of
different classes of target sets (e.g.\ upward-closed sets).
Our results concern both the Markov-chain model, where processes are scheduled
randomly, and the MDP model, where processes are picked by a scheduler.

\end{abstract}

\section{Introduction} \label{sec-introduction}

We study probabilistic basic parallel processes (pBPP), which is a stochastic model for concurrent systems with unbounded process spawning.
Processes can be of different types,
 and each type has a fixed probability distribution for generating new sub-processes.
A pBPP can be described using a notation similar to that of stochastic context-free grammars. For instance,
 \begin{align*}
   & X  \btran{0.2}  X X  \qquad X \btran{0.3}  X Y \qquad  X \btran{0.5} \varepsilon \qquad \qquad  Y  \btran{0.7}  X    \qquad Y \btran{0.3}  Y
 \end{align*}
describes a system with two types of processes.
Processes of type~$X$ can generate two processes of type~$X$, one process of each type, or zero processes with probabilities
$0.2$, $0.3$, and  $0.5$, respectively.
Processes of type~$Y$ can generate one process, of type $X$ or~$Y$, with probability $0.7$ and~$0.3$.
The order of processes on the right-hand side of each rule is not important.
Readers familiar with process algebra will identify this notation as a probabilistic version of
Basic Parallel Processes (BPPs), which is widely studied in automated verification,
 see e.g.\ \cite{EK95,HirshfeldJM96,Esp97,JancarBPP03,HuttelKS09,FroschleJLS10},

A \emph{configuration} of a pBPP indicates, for each type~$X$, how many processes of type~$X$ are present.
Writing $\Gamma$ for the finite set of types, a configuration is thus an element of~$\N^\Gamma$.
In a configuration $\alpha \in \N^\Gamma$ with $\alpha(X) \ge 1$ an $X$-process may be scheduled.
Whenever a process of type~$X$ is scheduled, a rule with $X$ on the left-hand side is picked randomly according
 to the probabilities of the rules, and then an $X$-process is replaced by processes as on the right-hand side.
In the example above, if an $X$-process is scheduled, then with probability~$0.3$ it is replaced by a new $X$-process and by a new $Y$-process.
This leads to a new configuration, $\alpha'$, with $\alpha'(X) = \alpha(X)$ and $\alpha'(Y) = \alpha(Y) + 1$.

\emph{Which} type is scheduled in a configuration~$\alpha \in \N^\Gamma$ depends on the model under consideration.
One possibility is that the type to be scheduled is selected randomly among those types~$X$ with $\alpha(X) \ge 1$.
In this way, a pBPP induces an (infinite-state) Markov chain.
We consider two versions of this Markov chain: in one version the type to be scheduled is picked
 using a uniform distribution on those types with at least one waiting process;
 in the other version the type is picked using a uniform distribution on the waiting processes.
For instance, in configuration~$\alpha$ with $\alpha(X) = 1$ and $\alpha(Y) = 2$,
 according to the ``type'' version, the probability of scheduling~$X$ is $1/2$,
 whereas in the ``process'' version, the probability is $1/3$.
Both models seem to make equal sense, so we consider them both in this paper.
As it turns out their difference is unimportant for our results.

In many contexts (e.g. probabilistic distributed protocols --- see
\cite{Nor04,LSS94}), it is more natural that this scheduling decision is not
taken randomly, but by a \emph{scheduler}.
Then the pBPP induces a Markov decision process (MDP), where a scheduler picks a type~$X$ to be scheduled,
 but the rule with $X$ on the left-hand side is selected probabilistically according to the probabilities on the rules.

In this paper we provide decidability results concerning \emph{coverability} with probability~$1$, or ``almost-sure'' coverability,
 which is a fundamental qualitative property of pBPPs.
We say a configuration~$\beta \in \N^\Gamma$ \emph{covers} a configuration~$\phi \in \N^\Gamma$
 if $\beta \ge \phi$ holds, where $\mathord{\ge}$ is meant componentwise.
For instance, $\phi$ may model a configuration with one producer and one consumer;
 then $\beta \ge \phi$ means that a transaction between a producer and a consumer can take place.
Another example is a critical section that can be entered only when a lock is obtained.
Given a pBPP, an initial configuration~$\alpha$, and target configurations $\phi_1, \ldots, \phi_k$,
 the \emph{coverability problem} asks whether with probability~$1$ it is the case that starting from~$\alpha$
 a configuration $\beta$ is reached that covers some~$\phi_i$.
One can equivalently view the problem as almost-sure reachability of an upward-closed set.

In Section~\ref{sec-markov-chain} we show using a Karp-Miller-style construction that the coverability problem for pBPP Markov chains is decidable.
We provide a nonelementary lower complexity bound.
In Section~\ref{sec-mdp} we consider the coverability problem for MDPs.
There the problem appears in two flavours, depending on whether the scheduler is ``angelic'' or ``demonic''.
In the angelic case we ask whether there \emph{exists} a scheduler so that a target is almost-surely covered.
We show that this problem is decidable, and if such a scheduler does exist one can synthesize one.
In the demonic case we ask whether a target is almost-surely covered, no matter what the scheduler (an operating system, for instance) does.
For the question to make sense we need to exclude unfair schedulers, i.e., those that never schedule a waiting process.
Using a robust fairness notion ($k$-fairness), which does not depend on the exact probabilities in the rules,
 we show that the demonic problem is also decidable.
In Section~\ref{sec-ptime} we show for the Markov chain and for both versions of the MDP problem
 that the coverability problem becomes P-time solvable, if the target configurations~$\phi_i$ consist of only one process each (i.e., are unit vectors).
Such target configurations naturally arise in concurrent systems (e.g. freedom
from deadlock: whether \emph{at least one} process eventually goes into
a critical section).
Finally, in Section~\ref{sec-semilinear} we show that the almost-sure reachability problem for semilinear sets,
 which generalizes the coverability problem, is undecidable for pBPP Markov chains and MDPs.
Some missing proofs can be found \iftechrep{in the appendix.}{in~\cite{BKL14-fossacs-report}.}

\paragraph{Related work.}

(Probabilistic) BPPs can be viewed as (stochastic) Petri nets
 where each transition has exactly one input place.
Stochastic Petri nets, in turn, are equivalent to probabilistic vector addition systems with states (pVASSs),
 whose reachability and coverability problems were studied in~\cite{AbdullaHM07}.
This work is close to ours; in fact, we build on fundamental results of~\cite{AbdullaHM07}.
Whereas we show that coverability for the Markov chain induced by a pBPP is decidable,
 it is shown in~\cite{AbdullaHM07} that the problem is undecidable for general pVASSs.
In~\cite{AbdullaHM07} it is further shown for general pVASSs that coverability becomes decidable
 if the target sets are ``$Q$-states''.
If we apply the same restriction on the target sets, coverability becomes polynomial-time decidable for pBPPs, see Section~\ref{sec-ptime}.
MDP problems are not discussed in~\cite{AbdullaHM07}.

The MDP version of pBPPs was studied before under the name \emph{task systems}~\cite{BEKL12:IandC}.
There, the scheduler aims at a ``space-efficient'' scheduling, which is one where the maximal number of processes is minimised.
Goals and techniques of this paper are very different from ours.

Certain classes of non-probabilistic 2-player games on Petri nets were studied in~\cite{Raskin2005}.
Our MDP problems can be viewed as games between two players, Scheduler and Probability.
One of our proofs (the proof of Theorem~\ref{thm-mdp-universal}) is inspired by proofs in~\cite{Raskin2005}.

The notion of $k$-fairness that we consider in this paper is not new. Similar
notions have appeared in the literature of concurrent systems under the
name of ``bounded fairness'' (e.g. see \cite{BF} and its citations).

\section{Preliminaries} \label{sec-preliminaries}

We write $\N = \{0, 1, 2, \ldots\}$.
For a countable set $X$ we write $\dist(X)$ for the set of \emph{probability distributions} over~$X$;
 i.e., $\dist(X)$ consists of those functions $f : X \to [0,1]$ such that $\sum_{x \in X} f(x) = 1$.

\paragraph{Markov Chains.} A \emph{Markov chain} is a pair $\DTMC = (Q,\delta)$,
where $Q$ is a countable (finite or infinite) set of states, and $\delta: Q \to \dist(Q)$ is a probabilistic transition function
 that maps a state to a probability distribution over the successor states.
Given a Markov chain we also write $s \tran{p} t$ or $s \tran{} t$ to indicate that $p = \delta(s)(t) > 0$.
A \emph{run} is an infinite sequence $s_0 s_1 \cdots \in Q^\omega$ with $s_i \tran{} s_{i+1}$ for $i \in \N$.
We write $\Run(s_0 \cdots s_k)$ for the set of runs that start with $s_0 \cdots s_k$.
To every initial state $s_0 \in S$ we associate the probability space $(\Run(s_0),\F,\P)$
where $\F$ is the $\sigma$-field generated by all basic cylinders $\Run(s_0 \cdots s_k)$ with $s_0 \cdots s_k \in Q^*$,
and $\P: \F \to [0,1]$ is the unique probability measure such that $\P(\Run(s_0 \cdots s_k)) = \prod_{i=1}^{k} \delta(s_{i-1})(s_i)$.
For a state $s_0 \in Q$ and a set $F \subseteq Q$, we write $s_0 \models \Diamond F$ for the event that a run started in~$s_0$ hits~$F$.
Formally, $s_0 \models \Diamond F$ can be seen as the set of runs $s_0 s_1 \cdots$ such that there is $i \ge 0$ with $s_i \in F$.
Clearly we have $\P(s_0 \models \Diamond F) > 0$ if and only if in~$\DTMC$ there is a path from~$s_0$ to a state in~$F$.
Similarly, for $Q_1, Q_2 \subseteq Q$ we write $s_0 \models Q_1 \U Q_2$ to denote the set of runs $s_0 s_1 \cdots$
 such that there is $j \ge 0$ with $s_j \in Q_2$ and $s_i \in Q_1$ for all $i < j$.
We have $\P(s_0 \models Q_1 \U Q_2) > 0$ if and only if in~$\DTMC$ there is a path from $s_0$ to a state in~$Q_2$ using only states in~$Q_1$.
A Markov chain is \emph{globally coarse} with respect to a set $F \subseteq Q$ of configurations,
 if there exists $c > 0$ such that for all $s_0 \in Q$ we have that $\P(s_0 \models \Diamond F) > 0$ implies $\P(s_0 \models \Diamond F) \ge c$.

\paragraph{Markov Decision Processes.}
A \emph{Markov decision process (MDP)} is a tuple $\D = (Q, A, \En, \delta)$,
 where $Q$ is a countable set of states, $A$ is a finite set of actions,
 $\En : Q \to 2^A \setminus \emptyset$ is an action enabledness function that assigns to each state~$s$ the set $\En(s)$ of actions enabled in~$s$,
 and $\delta : S \times A \to \dist(S)$ is a probabilistic transition function that maps a state~$s$ and an action $a \in \En(s)$ enabled in~$s$
 to a probability distribution over the successor states.
A \emph{run} is an infinite alternating sequence of states and actions $s_0 a_1 s_1 a_2 \cdots$
 such that for all $i \ge 1$ we have $a_i \in \En(s_{i-1})$ and $\delta(s_{i-1}, a_i)(s_i) > 0$.
For a finite word $w = s_0 a_1 \cdots s_{k-1} a_k s_k \in Q (A Q)^*$ we write $\last(w) = s_k$.
A \emph{scheduler} for~$\D$ is a function $\sigma : Q (A Q)^* \to \dist(A)$
 that maps a run prefix $w \in Q (A Q)^*$, representing the history of a play,
 to a probability distribution over the actions enabled in~$\last(w)$.
We write $\Run(w)$ for the set of runs that start with $w \in Q (A Q)^*$.
To an initial state $s_0 \in S$ and a scheduler~$\sigma$ we associate the probability space $(\Run(s_0),\F,\P_\sigma)$,
 where $\F$ is the $\sigma$-field generated by all basic cylinders $\Run(w)$ with $w \in \{s_0\} (A Q)^*$,
 and $\P_\sigma: \F \to [0,1]$ is the unique probability measure such that
 $\P(\Run(s_0)) = 1$, and $\P(\Run(w a s)) = \P(\Run(w)) \cdot \sigma(w)(a) \cdot \delta(\last(w),a)(s)$
  for all $w \in \{s_0\} (A Q)^*$ and all $a \in A$ and all $s \in Q$.
A scheduler~$\sigma$ is called \emph{deterministic} if for all $w \in Q (A Q)^*$ there is $a \in A$ with $\sigma(w)(a) = 1$.
A scheduler~$\sigma$ is called \emph{memoryless} if for all $w, w' \in Q (A Q)^*$ with $\last(w) = \last(w')$ we have $\sigma(w) = \sigma(w')$.
When specifying events, i.e., measurable subsets of $\Run(s_0)$, the actions are often irrelevant.
Therefore, when we speak of runs $s_0 s_1 \cdots$ we mean the runs $s_0 a_1 s_1 a_2 \cdots$ for arbitrary $a_1, a_2, \ldots \in A$.
E.g., in this understanding we view $s_0 \models \Diamond F$ with $s_0 \in Q$ and $F \subseteq Q$ as an event.


\paragraph{Probabilistic BPPs and their configurations.}
A \emph{probabilistic BPP (pBPP)} is a tuple $\S = (\Gamma,\mathord{\btran{}},\Prob)$,
 where $\Gamma$ is a finite set of \emph{types},
 $\mathord{\btran{}} \subseteq \Gamma \times \N^\Gamma$ is a finite set of \emph{rules}
  such that for every $X \in \Gamma$ there is at least one rule of the form $X \btran{} \alpha$,
 and $\Prob$ is a function that to every rule \mbox{$X \btran{} \alpha$} assigns its
  probability \mbox{$\Prob(X \btran{} \alpha) \in (0,1] \cap \Q$} so that for
  all $X \in \Gamma$ we have $\sum_{X \btran{} \alpha} \Prob(X \btran{} \alpha) = 1$.
We write $X \btran{p} \alpha$ to denote that $\Prob(X \btran{} \alpha) = p$.
A \emph{configuration} of~$\S$ is an element of $\N^\Gamma$.
We write $\alpha_1 + \alpha_2$ and $\alpha_1 - \alpha_2$ for componentwise addition and subtraction of two configurations $\alpha_1, \alpha_2$.
When there is no confusion, we may identify words $u \in \Gamma^*$ with the configuration $\alpha \in \N^\Gamma$
 such that for all $X \in \Gamma$ we have that $\alpha(X) \in \N$ is the number of occurrences of~$X$ in~$u$.
For instance, we write $X X Y$ or $X Y X$ for the configuration~$\alpha$ with $\alpha(X) = 2$ and $\alpha(Y) = 1$
 and $\alpha(Z) = 0$ for $Z \in \Gamma \setminus \{X,Y\}$.
In particular, we may write~$\varepsilon$ for $\alpha$ with $\alpha(X) = 0$ for all $X \in \Gamma$.
For configurations $\alpha, \beta$ we write $\alpha \le \beta$ if $\alpha(X) \le \beta(X)$ holds for all $X \in \Gamma$;
 we write $\alpha < \beta$ if $\alpha \le \beta$ but $\alpha \ne \beta$.
For a configuration~$\alpha$ we define the \emph{number of types} $\Wt{\alpha} = |\{X \in \Gamma \mid \alpha(X) \ge 1\}|$
 and the \emph{number of processes} $\Wp{\alpha} = \sum_{X \in \Gamma} \alpha(X)$.
Observe that we have $\Wt{\alpha} \le \Wp{\alpha}$.
A set $F \subseteq \N^\Gamma$ of configurations is called \emph{upward-closed} (\emph{downward-closed}, respectively) if for all $\alpha \in F$
 we have that $\alpha \le \beta$ implies $\beta \in F$ ($\alpha \ge \beta$ implies $\beta \in F$, respectively).
For $\alpha \in \N^\Gamma$ we define $\alpha \up := \{\beta \in \N^\Gamma \mid \beta \ge \alpha\}$.
For $F \subseteq \N^\Gamma$ and $\alpha \in F$ we say that $\alpha$ is a \emph{minimal element} of~$F$, if there is no $\beta \in F$ with $\beta < \alpha$.
It follows from Dickson's lemma that every upward-closed set has finitely many minimal elements;
 i.e., $F$ is upward-closed if and only if $F = \phi_1 \up \cup \ldots \cup \phi_n \up$ holds for some $n \in \N$ and $\phi_1, \ldots, \phi_n \in \N^\Gamma$.

\paragraph{Markov Chains induced by a pBPP.}
To a pBPP~$\S = (\Gamma,\mathord{\btran{}},\Prob)$ we associate the Markov chains $\Mt(\S) = (\N^\Gamma, \deltat)$ and $\Mp(\S) = (\N^\Gamma, \deltap)$
 with $\deltat(\varepsilon,\varepsilon) = \deltap(\varepsilon,\varepsilon) = 1$ and for $\alpha \ne \varepsilon$
\[
 \deltat(\alpha,\gamma) = \mathop{\sum_{X \btran{p} \beta \text{ s.t.\ } \alpha(X) \ge 1}}_{\text{ and } \gamma = \alpha - X + \beta} \frac{p}{\Wt{\alpha}}
 \quad \text{and} \quad
 \deltap(\alpha,\gamma) = \mathop{\sum_{X \btran{p} \beta \text{ s.t.\ }}}_{\gamma = \alpha - X + \beta}
  \frac{\alpha(X) \cdot p}{\Wp{\alpha}} \,.
\]
In words, the new configuration~$\gamma$ is obtained from~$\alpha$ by replacing an $X$-process
 with a configuration randomly sampled according to the rules $X \btran{p} \beta$.
In~$\Mt(\S)$ the selection of~$X$ is based on the number of types in~$\alpha$,
 whereas in~$\Mp(\S)$ it is based on the number of processes in~$\alpha$.
We have $\deltat(\alpha,\gamma) = 0$ iff $\deltap(\alpha,\gamma) = 0$.
We write $\Pt$ and $\Pp$ for the probability measures in $\Mt(\S)$ and $\Mp(\S)$, respectively.

\paragraph{The MDP induced by a pBPP.}
To a pBPP~$\S = (\Gamma,\mathord{\btran{}},\Prob)$ we associate the MDP $\D(\S) = (\N^\Gamma, \Gamma \cup \{\bot\}, \En, \delta)$
 with a fresh action $\bot \not\in \Gamma$, and $\En(\alpha) = \{X \in \Gamma \mid \alpha(X) \ge 1\}$ for $\varepsilon \ne \alpha \in \N^\Gamma$
 and $\En(\varepsilon) = \{\bot\}$,
 and $\delta(\alpha,X)(\alpha - X + \beta) = p$ whenever $\alpha(X) \ge 1$ and $X \btran{p} \beta$,
 and $\delta(\varepsilon, \bot)(\varepsilon) = 1$.
As in the Markov chain, the new configuration~$\gamma$ is obtained from~$\alpha$ by replacing an $X$-process
 with a configuration randomly sampled according to the rules $X \btran{p} \beta$.
But in contrast to the Markov chain the selection of~$X$ is up to a scheduler.

\section{The Coverability Problem for the Markov Chain} \label{sec-markov-chain}

In this section we study the \emph{coverability problem} for the Markov chains induced by a pBPP.
We say a run $\alpha_0 \alpha_1 \cdots$ of a pBPP $\S = (\Gamma,\mathord{\btran{}},\Prob)$ \emph{covers} a configuration $\phi \in \N^\Gamma$,
 if $\alpha_i \ge \phi$ holds for some $i \in \N$.
The coverability problem asks whether it is almost surely the case that some configuration from a finite set $\{\phi_1, \ldots, \phi_n\}$ will be covered.
More formally, the coverability problem is the following.
Given a pBPP $\S = (\Gamma,\mathord{\btran{}},\Prob)$, an initial configuration $\alpha_0 \in \N^\Gamma$,
 and finitely many configurations $\phi_1, \ldots, \phi_n$,
 does $\Pt(\alpha_0 \models \Diamond F) = 1$ hold, where $F = \phi_1 \up \cup \ldots \cup \phi_n \up$?
Similarly, does $\Pp(\alpha_0 \models \Diamond F) = 1$ hold?
We will show that those two questions always have the same answer.

In Section~\ref{sub-markov-decidability} we show that the coverability problem is decidable.
In Section~\ref{sub-markov-lower-bound} we show that the complexity of the coverability problem is nonelementary.

\subsection{Decidability} \label{sub-markov-decidability}

For deciding the coverability problem we use the approach of~\cite{AbdullaHM07}.
The following proposition is crucial for us:

\begin{proposition} \label{prop-Abdulla}
Let $\DTMC = (Q,\delta)$ be a Markov chain and $F \subseteq Q$ such that $\DTMC$ is globally coarse with respect to~$F$.
Let $\bar{F} = Q \setminus F$ be the complement of~$F$
 and let $\widetilde{F} := \{s \in Q \mid \P(s \models \Diamond F) = 0\} \subseteq \bar{F}$
 denote the set of states from which $F$ is not reachable in~$\DTMC$.
Let $s_0 \in Q$.
Then we have $\P(s_0 \models \Diamond F) = 1$ if and only if $\P(s_0 \models \bar{F} \U \widetilde{F}) = 0$.
\end{proposition}
\begin{proof}
Immediate from~\cite[Lemmas 3.7, 5.1 and 5.2]{AbdullaHM07}.
\qed
\end{proof}
In other words, Proposition~\ref{prop-Abdulla} states that $F$ is almost surely reached
 if and only if there is no path to~$\widetilde{F}$ that avoids~$F$.
Proposition~\ref{prop-Abdulla} will allow us to decide the coverability problem by computing only reachability relations in~$\DTMC$, ignoring the probabilities.

Recall that for a pBPP $\S = (\Gamma,\mathord{\btran{}},\Prob)$, the Markov chains $\Mt(\S)$ and $\Mp(\S)$ have the same structure;
 only the transition probabilities differ.
In particular, if $F \subseteq \N^\Gamma$ is upward-closed, the set $\widetilde{F}$, as defined in Proposition~\ref{prop-Abdulla},
 is the same for $\Mt(\S)$ and $\Mp(\S)$.
Moreover, we have the following proposition (full proof \iftechrep{in Appendix~\ref{app-markov-chain}}{in~\cite{BKL14-fossacs-report}}).
\newcommand{\stmtpropgloballycoarse}{
 Let $\S = (\Gamma,\mathord{\btran{}},\Prob)$ be a pBPP.
 Let $F \subseteq \N^\Gamma$ be upward-closed.
 Then the Markov chains $\Mt(\S)$ and $\Mp(\S)$ are globally coarse with respect to~$F$.
}
\begin{proposition} \label{prop-globally-coarse}
 \stmtpropgloballycoarse
\end{proposition}
\begin{proof}[sketch]
The statement about~$\Mt(\S)$ follows from~\cite[Theorem~4.3]{AbdullaHM07}.
For the statement about~$\Mp(\S)$ it is crucial to argue that starting with any configuration~$\alpha \in \N^\Gamma$
 it is the case with probability~$1$ that every type~$X$ with $\alpha(X) \ge 1$ is eventually scheduled.
Since $F$ is upward-closed it follows that for all $\alpha, \beta \in \N^\Gamma$ with $\alpha \le \beta$
 we have $\Pp(\alpha \models \Diamond F) \le \Pp(\beta \models \Diamond F)$.
Then the statement follows from Dickson's lemma.

For an illustration of the challenge, consider the pBPP with $X \btran{1} X X$ and $Y \btran{1} Y Y$, and let $F = X X \up$.
Clearly we have $\Pp(X \models \Diamond F) = 1$, as the $X$-process is scheduled immediately.
Now let $\alpha_0 = X Y$.
Since $\alpha_0 \ge X$, the inequality claimed above implies $\Pp(\alpha_0 \models F) = 1$.
Indeed, the probability that the $X$-process in~$\alpha_0$ is \emph{never} scheduled is at most $\frac12 \cdot \frac23 \cdot \frac 34 \cdot \ldots$,
 which is~$0$.
Hence $\Pp(\alpha_0 \models \Diamond F) = 1$.
\qed
\end{proof}

The following proposition 
 follows by combining Propositions \ref{prop-Abdulla}~and~\ref{prop-globally-coarse}.
\newcommand{\stmtpropqualreachequal}{
Let $\S = (\Gamma,\mathord{\btran{}},\Prob)$ be a pBPP.
Let $F \subseteq \N^\Gamma$ be upward-closed.
Let $\alpha_0 \in \N^\Gamma$.
We have: 
\begin{align*}
                          \Pt(\alpha_0 \models \Diamond F) = 1
  & \Longleftrightarrow\  \Pt(\alpha_0 \models \bar{F} \U \widetilde{F}) = 0 \\
  & \Longleftrightarrow\  \Pp(\alpha_0 \models \bar{F} \U \widetilde{F}) = 0
    \Longleftrightarrow\  \Pp(\alpha_0 \models \Diamond F) = 1
\end{align*}
}
\begin{proposition} \label{prop-qual-reach-equal}
 \stmtpropqualreachequal
\end{proposition}
By Proposition~\ref{prop-qual-reach-equal} we may in the following omit the subscript from $\Pt,\Pp, \Mt, \Mp$ if it does not matter.
We have the following theorem.
\begin{theorem} \label{thm-qual-reach-decid}
The coverability problem is decidable:
given a pBPP $\S = (\Gamma,\mathord{\btran{}},\Prob)$, an upward-closed set $F \subseteq \N^\Gamma$, and a configuration $\alpha_0 \in \N^\Gamma$,
 it is decidable whether $\P(\alpha_0 \models \Diamond F) = 1$ holds.
\end{theorem}
\newcommand{\stmtlemqualreachdecid}{
 Let $\alpha_0 \in \bar{F}$ and let $\gamma \in \N^\Gamma$.
 Let $\alpha_0 \to \alpha_1 \to \ldots \to \alpha_k$ be a shortest path in~$\DTMC(\S)$
  such that $\alpha_0, \ldots, \alpha_k \in \bar{F}$ and $\alpha_k \le \gamma$.
 Then for all $i,j$ with $0 \le i < j \le k$ we have $\alpha_i \not\le \alpha_j$.
}
\begin{proof}
The complement of~$\widetilde{F}$ (i.e., the set of configurations from which $F$ is reachable) is upward-closed,
 and its minimal elements can be computed by a straightforward fixed-point
computation (this is even true for the more general model of pVASS, e.g., see
\cite[Remark 4.2]{AbdullaHM07}).
By Proposition~\ref{prop-qual-reach-equal} it suffices to decide whether $\P(\alpha_0 \models \bar{F} \U \widetilde{F}) > 0$ holds.
Define
$
 R := \{\alpha \in \bar{F} \mid \alpha \text{ is reachable from $\alpha_0$ via $\bar{F}$-configurations}\}.
$
Observe that $\P(\alpha_0 \models \bar{F} \U \widetilde{F}) > 0$ if and only if
$R \cap \widetilde{F} \ne \emptyset$.
We can now give a Karp-Miller-style algorithm for checking that
$R \cap \widetilde{F} \ne \emptyset$:
(i) Starting from~$\alpha_0$, build a tree of configurations reachable
from~$\alpha_0$ via $\bar{F}$-configurations (i.e., at no stage
$F$-configurations are added to this tree) --- for example, in a breadth-first
search manner --- but stop expanding a leaf node $\alpha_k$ as soon as we
discover that the branch
$\alpha_0 \to \cdots \to \alpha_k$ satisfies the following: $\alpha_j \le
\alpha_k$ for some $j < k$.
(ii) As soon as a node $\alpha \in \widetilde{F}$
is generated, terminate and output ``yes''.
(iii) When the tree construction is
completed without finding nodes in $\widetilde{F}$, terminate and output ``no''.

To prove correctness of the above algorithm, we first prove
\underline{termination}.
To this end, it suffices to show that the constructed
tree is finite. To see this, observe first that every branch in the constructed
tree is of finite length. This is an immediate consequence of Dickon's lemma
and our policy of terminating a leaf node $\alpha$ that satisfies $\alpha'
\le \alpha$, for some ancestor $\alpha'$ of $\alpha$ in this tree. Now
since all branches of the tree are finite, K\"{o}nig's lemma shows that the
tree itself must be finite (since each node has finite degree).

To prove \underline{partial correctness}, it suffices to show that the policy
of terminating a leaf node $\alpha$ that satisfies $\alpha'
\le \alpha$, for some ancestor $\alpha'$ of $\alpha$ in this tree, is
valid.
That is, we want to show that if $R \cap \widetilde{F} \neq \emptyset$
then a witnessing vector $\gamma \in R \cap \widetilde{F}$ will be found by the
algorithm.
We have the following lemma whose proof is \iftechrep{in Appendix~\ref{app-markov-chain}}{in~\cite{BKL14-fossacs-report}}.
\begin{lemma} \label{lem-qual-reach-decid}
 \stmtlemqualreachdecid
\end{lemma}
Let $R \cap \widetilde{F} \neq \emptyset$
 and let $\gamma \in \N^\Gamma$ be a minimal element of $R \cap \widetilde{F}$.
By Lemma~\ref{lem-qual-reach-decid} our algorithm does not prune any shortest path
 from~$\alpha_0$ to~$\gamma$.
Hence it outputs ``yes''.
\qed
\end{proof}

\subsection{Nonelementary Lower Bound} \label{sub-markov-lower-bound}

We have the following lower-bound result:

\newcommand{\stmtthmnonelementary}{
 The complexity of the coverability problem is nonelementary.
}
\begin{theorem} \label{thm-nonelementary}
 \stmtthmnonelementary
\end{theorem}
The proof is technically involved.
\begin{proof}[sketch]
We claim that there exists a nonelementary function $f$ such that given a
2-counter machine $M$ running in space $f(k)$, we can compute a pBPP
$\S = (\Gamma, \mathord{\btran{}}, \Prob)$ of size $\leq k$,
an upward-closed set $F \subseteq \N^\Gamma$
(with at most $k$ minimal elements, described by numbers at most~$k$),
and a type~$X_0 \in \Gamma$,
such that $\P(X_0 \models \Diamond F) = 1$ holds if and only if $M$ does not terminate.
Recall that by Proposition~\ref{prop-qual-reach-equal} we have that  $\P(X_0 \models \Diamond F) = 1$ is equivalent to
$\P(X_0 \models \bar{F} \U \widetilde{F}) = 0$.

Since the exact values of the probabilities do not matter,
 it suffices to construct a (nonprobabilistic) BPP~$\S$.
Further, by adding processes that can spawn everything (and hence cannot take part in $\widetilde{F}$-configurations)
one can change the condition of reaching
$\widetilde{F}$ to reaching a downward closed set $G \subseteq \bar{F}$.
So the problem we are reducing to is:
 does there exist a path in~$\S$ that is contained in $\bar{F}$ and goes
 from $X_0$ to a downward closed set~$G$.

By defining $F$ suitably we can add various restrictions on the
behaviour of our BPP. For example, the following example allows $X$ to be turned
into $Y$ if and only if there is no $Z$ present:

$$
X \btran{} YW \qquad W \btran{} \eps \qquad F = WZ \up
$$

Doubling the number of a given process is straightforward, and it is also
possible to divide the number of a given process by two.
Looking only at runs inside $\bar{F}$, the
following BPP can turn all its $X$-processes into half as many
$X'$-processes. (Note that more $X'$-processes could be spawned, but because
of the monotonicity of the system, the ``best'' runs are those that spawn a
minimal number of processes.)

\[
\begin{matrix}
X \btran{} TP \qquad T \btran{} \overline{P} \qquad
P \btran{} \eps \qquad \overline{P} \btran{} \eps \\[3mm]
P_1 \btran{} \overline{P_2} \qquad \overline{P_2} \btran{} P_2 \qquad
P_2 \btran{} \overline{P_1} \qquad \overline{P_1} \btran{} P_1 X' \\[3mm]
F = P\overline{P_1} \up ~\cup~ P\overline{P_2} \up ~\cup~
    \overline{P}P_1 \up ~\cup~ \overline{P}P_2 \up ~\cup~ T^2 \up \\[3mm]
\alpha_{\mathit{init}} = X^n \overline{P_1}
\end{matrix}
\]

Let us explain this construction. In order to make an $X$-process disappear,
we need to create temporary processes $P$ and $\overline{P}$. However, these
processes are incompatible, respectively, with $\overline{P_i}$ and $P_i$. Thus,
destroying an $X$-process requires the process $P_1$ to move into
$\overline{P_1}$ and then into $P_2$. By repeatedly destroying $X$-processes,
this forces the creation of half as many $X'$-processes.

It is essential for our construction to have a loop-gadget that performs a cycle of processes
$A \btran{} B \btran{} C \btran{} A$ exactly $k$ times (``$k$-loop''). By activating/disabling
transitions based on the absence/presence of an $A$-, $B$- or $C$-process, we can
force an operation to be performed $k$ times. For example, assuming the construction of a $k$-loop gadget, the following BPP
doubles the number of $X$-processes $k$ times:

\[
\begin{matrix}
X \btran{} Y \qquad Y \btran{} ZZ \qquad Z \btran{} X \\[3mm]
\text{(rules for $k$-loop on $A/B/C$)} \\[3mm]
F = XB \up ~\cup~ YC \up ~\cup~ ZA \up
\end{matrix}
\]

For the loop to perform $A \btran{} B$, all $X$-processes
have to be turned into $Y$. Similarly, performing
$B \btran{} C \btran{} A$ requires the $Y$-processes to be turned into $Z$,
then into $X$. Thus, in order to perform one iteration of the loop, one needs
to double the number of $X$-processes.

To implement such a loop we need two more gadgets:
 one for creating $k$ processes, and one for consuming $k$ processes.
By turning a created process into a consumed process on at a time, we obtain the required cycle.
Here is an example:

\[
\begin{matrix}
I \btran{} A \qquad A \btran{} B \qquad B \btran{} C \qquad C \btran{} \eps
\\[3mm]
\overline{A} \btran{} \overline{B} \qquad \overline{B} \btran{} \overline{C}F
\qquad \overline{C} \btran{} A
\\[3mm]
(\text{rules for a gadget to consume $k$ processes $F$})
\\[3mm]
(\text{rules for a gadget to spawn $k$ processes $I$})
\\[3mm]
F = A\overline{A} \up ~\cup~ B\overline{B} \up ~\cup~ C\overline{C} \up
~\cup~ AA \up ~\cup~ BB \up ~\cup~ CC \up
\\[3mm]
\alpha_{\mathit{init}} = \overline{A}
\end{matrix}
\]

By combining a $k$-loop with a multiplier or a divider, we can spawn or
consume $2^k$ processes.
This allows us to create a $2^k$-loop.
By iterating this construction, we get a BPP of exponential size (each loop
requires two lower-level loops) that is able to spawn or consume $2^{2^{...^k}}$ processes.

It remains to simulate our 2-counter machine~$M$. The main idea is to spawn
an initial budget $b$ of processes, and to make sure that this number stays
the same along the run. Zero-tests are easy to implement; the
difficulty lies in the increments and decrements. The solution is to maintain,
for each simulated counter, two pools of processes $X$ and $\overline{X}$,
such that if the counter is supposed to have value $k$, then we have
processes $X^k \overline{X}^{b-k}$.
Now, incrementing consist in turning all these
processes into backup processes, except one $\overline{X}$-process. Then, we turn
this process into an $X$-process, and return all backup process to their
initial type.

\iftechrep{Appendix~\ref{sec-complexity} provides}{In~\cite{BKL14-fossacs-report} we provide}
complete details of the proofs, including graphical representations of the processes involved.
\qed
\end{proof}

\section{The Coverability Problem for the MDP} \label{sec-mdp}

In the following we investigate the controlled version of the pBPP model.
Recall from Section~\ref{sec-preliminaries} that a pBPP $\S = (\Gamma,\mathord{\btran{}},\Prob)$ induces an MDP $\D(\S)$
 where in a configuration $\varepsilon \ne \alpha \in \N^\Gamma$ a scheduler~$\sigma$ picks a type~$X$ with $\alpha(X) \ge 1$.
The successor configuration is then obtained randomly from~$\alpha$ according to the rules in~$\S$ with $X$ on the left-hand side.

We investigate (variants of) the decision problem that asks, given $\alpha_0 \in \N^\Gamma$ and an upward-closed set~$F \subseteq \N^\Gamma$,
 whether $\P_\sigma(\alpha_0 \models \Diamond F) = 1$ holds for some scheduler (or for all schedulers, respectively).

\subsection{The Existential Problem} \label{sub-existential}
In this section we consider the scenario where we ask for a scheduler
 that makes the system reach an upward-closed set with probability~$1$.
We prove the following theorem:
\newcommand{\stmtthmmdpexistential}{
 Given a pBPP~$\S = (\Gamma,\mathord{\btran{}},\Prob)$ and a configuration $\alpha_0 \in \N^\Gamma$ and an upward-closed set~$F \subseteq \N^\Gamma$,
  it is decidable whether there exists a scheduler~$\sigma$ with $\P_\sigma(\alpha_0 \models \Diamond F) = 1$.
 If such a scheduler exists, one can compute a deterministic and memoryless scheduler~$\sigma$ with $\P_\sigma(\alpha_0 \models \Diamond F) = 1$.
}
\begin{theorem} \label{thm-mdp-existential}
 \stmtthmmdpexistential
\end{theorem}
\begin{proof}[sketch]
The proof (\iftechrep{in Appendix~\ref{app-mdp}}{in~\cite{BKL14-fossacs-report}}) is relatively long.
The idea is to abstract the MDP~$\D(\S)$ (with $\N^\Gamma$ as state space) to an ``equivalent'' finite-state MDP.
The state space of the finite-state MDP is $Q := \{0, 1, \ldots, K\}^\Gamma \subseteq \N^\Gamma$,
 where $K$ is the largest number that appears in the minimal elements of~$F$.
For finite-state MDPs, reachability with probability~$1$ can be decided in polynomial time,
 and an optimal deterministic and memoryless scheduler can be synthesized.

When setting up the finite-state MDP, special care needs to be taken of transitions
 that would lead from~$Q$ to a configuration~$\alpha$ outside of~$Q$, i.e., $\alpha \in \N^\Gamma \setminus Q$.
Those transitions are redirected to a probability distribution on~$T_\alpha$ with $T_\alpha \subseteq Q$,
 so that each configuration in~$T_\alpha$ is ``equivalent'' to some configuration $\beta \in \N^\Gamma$
  that could be reached from~$\alpha$ in the infinite-state MDP~$\D(\S)$,
  if the scheduler follows a particular optimal strategy in~$\D(\S)$.
(One needs to show that indeed with probability~$1$ such a~$\beta$ is reached in the infinite-state MDP,
 if the scheduler acts according to this strategy.)
This optimal strategy is based on the observation that whenever in configuration $\beta \in \N^\Gamma$ with $\beta(X) > K$ for some~$X$,
 then type $X$ can be scheduled.
This is without risk, because after scheduling~$X$, at least $K$ processes of type~$X$ remain,
 which is enough by the definition of~$K$.
The benefit of scheduling such~$X$ is that processes appearing on the right-hand side of $X$-rules may be generated, possibly helping to reach~$F$.
For computing~$T_\alpha$, we rely on decision procedures for the reachability problem in Petri nets,
 which prohibits us from giving an upper complexity bound.
\qed
\end{proof}

\subsection{The Universal Problem} \label{sub-universal}

In this section we consider the scheduler as adversarial in the sense that it tries to avoid the upward-closed set~$F$.
We say ``the scheduler wins'' if it avoids~$F$ forever.
We ask if the scheduler can win with positive probability:
 given $\alpha_0$ and~$F$, do we have $\P_\sigma(\alpha_0 \models \Diamond F) = 1$ for all schedulers~$\sigma$?
For the question to make sense, we need to rephrase it, as we show now.
Consider the pBPP $\S = (\Gamma,\mathord{\btran{}},\Prob)$ with $\Gamma = \{X,Y\}$ and the rules $X \btran{1} X X$ and $Y \btran{1} Y Y$.
Let $F = X X \up$.
If $\alpha_0 = X$, then, clearly, we have $\P_\sigma(\alpha_0 \models \Diamond F) = 1$ for all schedulers~$\sigma$.
However, if $\alpha_0 = X Y$, then there is a scheduler~$\sigma$ with $\P_\sigma(\alpha_0 \models \Diamond F) = 0$:
 take the scheduler~$\sigma$ that always schedules~$Y$ and never~$X$.
Such a scheduler is intuitively \emph{unfair}.
If an operating system acts as a scheduler, a minimum requirement would be that waiting processes are scheduled eventually.

We call a run $\alpha_0 X_1 \alpha_1 X_2 \ldots$ in the MDP~$\D(\S)$ \emph{fair}
 if for all $i \ge 0$ and all $X \in \Gamma$ with $\alpha_i(X) \ge 1$ we have $X = X_j$ for some $j > i$.
We call a scheduler~$\sigma$ \emph{classically fair} if it produces only fair runs.
\begin{example}
 Consider the pBPP with $X \btran{1} Y$ and $Y \btran{0.5} Y$ and $Y \btran{0.5} X$.
 Let $F = Y Y \up$.
 Let $\alpha_0 = X X$.
 In configuration $\alpha = X Y$ the scheduler has to choose between two options:
 It can pick~$X$, resulting in the successor configuration $Y Y \in F$, which is a ``loss'' for the scheduler.
 Alternatively, it picks~$Y$, which results in $\alpha$ or~$\alpha_0$, each with probability~$0.5$.
 If it results in~$\alpha$, nothing has changed; if it results in~$\alpha_0$, we say a ``a round is completed''.
 Consider the scheduler~$\sigma$ that acts as follows.
 When in configuration $\alpha = X Y$ and in the $i$th round,
  it picks~$Y$ until either the next round (the $(i+1)$st round) is completed or $Y$ has been picked $i$ times in this round.
 In the latter case it picks~$X$ and thus loses.
 Clearly, $\sigma$ is classically fair (provided that it behaves in a classically fair way after it loses, for instances using round-robin).
 The probability of losing in the $i$th round is $2^{-i}$.
 Hence the probability of losing is $\P_\sigma(\alpha_0 \models \Diamond F) = 1 - \prod_{i=1}^\infty (1 - 2^{-i}) < 1$.
 (For this inequality, recall that for a sequence $(a_i)_{i \in \N}$ with $a_i \in (0,1)$ we have $\prod_{i \in \N} (1-a_i) = 0$ if and only if
  the series $\sum_{i \in \N} a_i$ diverges.)
 One can argue along these lines that any classically fair scheduler needs to play longer and longer rounds
  in order to win with positive probability.
 In particular, such schedulers need infinite memory.
\end{example}
It is hardly conceivable that an operating system would ``consider'' such schedulers.
Note that the pBPP from the previous example has a finite state space.

In the probabilistic context, a commonly used alternative notion is \emph{probabilistic fairness},
 see e.g.~\cite{Hart1983,Vardi85} or \cite{deAlfaro1999} for an overview (the term \emph{probabilistic fairness}
  is used differently in~\cite{deAlfaro1999}).
We call a scheduler~$\sigma$ \emph{probabilistically fair} if with probability~$1$ it produces a fair run.
\begin{example}
 For the pBPP from the previous example, consider the scheduler~$\sigma$ that picks~$Y$ until the round is completed.
 Then $\P_\sigma(\alpha \models \Diamond F) = 0$ and $\sigma$ is probabilistically fair.
\end{example}
The following example shows that probabilistic fairness for pBPPs can be unstable with respect to perturbations in the probabilities.
\begin{example} \label{ex-gambler-ruin}
 Consider a pBPP with
 \[
  X \btran{1} Y \qquad Y \btran{1} X Z \qquad Z \btran{p} Z Z \qquad Z \btran{1-p} \varepsilon \qquad \text{for some $p \in (0,1)$}
 \]
 and $F = Y Z \up$ and $\alpha_0 = X Z$.

 Let $p \le 0.5$.
 Then, by an argument on the ``gambler's ruin problem'' (see e.g.~\cite[Chapter XIV]{FellerVol1}),
 with probability~$1$ each $Z$-process produces only finitely many other $Z$-processes in its ``subderivation tree''.
 Consider the scheduler~$\sigma$ that picks~$Z$ as long as there is a $Z$-process.
 With probability~$1$ it creates a run of the following form:
 \[
   (X Z) \cdots (X) (Y) (X Z) \cdots (X) (Y) (X Z) \cdots (X) (Y) (X Z) \ldots
 \]
 Such runs are fair, so $\sigma$ is probabilistically fair and wins with probability~$1$.

 Let $p > 0.5$.
 Then, by the same random-walk argument, with probability~$1$ some $Z$-process (i.e., at least one of the $Z$-processes created by~$Y$)
  produces infinitely many other $Z$-processes.
 So any probabilistically fair scheduler~$\sigma$ produces, with probability~$1$, a $Y$-process before all $Z$-processes are gone,
  and thus loses.

 We conclude that a probabilistically fair scheduler~$\sigma$ with $\P_\sigma(\alpha_0 \models \Diamond F) < 1$ exists if and only if $p \le 0.5$.
\end{example}
The example suggests that deciding whether there exists a probabilistically fair scheduler~$\sigma$ with $\P_\sigma(\alpha_0 \models \Diamond F) < 1$
 requires arguments on (in general) multidimensional random walks.
In addition, the example shows that probabilistic fairness is not a robust notion when the exact probabilities are not known.

We aim at solving those problems by considering a stronger notion of fair runs:
Let $k \in \N$.
We call a run $\alpha_0 X_1 \alpha_1 X_2 \ldots$ \emph{$k$-fair}
 if for all $i \ge 0$ and all $X \in \Gamma$ with $\alpha_i(X) \ge 1$ we have that $X \in \{X_{i+1}, X_{i+2}, \ldots, X_{k}\}$.
In words, if $\alpha_i(X) \ge 1$, the type~$X$ has to be scheduled within time~$k$.
We call a scheduler \emph{$k$-fair} if it produces only $k$-fair runs.

\newcommand{\stmtthmmdpuniversal}{
 Given a pBPP $\S = (\Gamma,\mathord{\btran{}},\Prob)$, an upward-closed set~$F$, a number $k \in \N$, and a configuration $\alpha_0 \in \N^\Gamma$,
  it is decidable whether for all $k$-fair schedulers~$\sigma$ we have $\P_\sigma(\alpha_0 \models \Diamond F) = 1$.
}
\begin{theorem} \label{thm-mdp-universal}
 \stmtthmmdpuniversal
\end{theorem}
The proof is inspired by proofs in~\cite{Raskin2005}, and combines new insights with the technique of Theorem~\ref{thm-qual-reach-decid},
 \iftechrep{see Appendix~\ref{app-mdp}}{see~\cite{BKL14-fossacs-report}}.
We remark that the proof shows that the exact values of the positive probabilities do not matter.

\section{$Q$-States Target Sets} \label{sec-ptime}

In this section, we provide a sensible restriction of input target sets which
yields polynomial-time solvability of our problems.
Let $Q = \{X_1, \ldots, X_n\} \subseteq \Gamma$.
The \defn{$Q$-states set} is the upward-closed set $F = X_1 \up \cup \ldots \cup X_n \up$.
There are two reasons to
consider $Q$-states target sets.
Firstly, $Q$-states target sets are sufficiently
expressive to capture common examples in the literature of distributed
protocols, e.g., freedom from deadlock and resource starvation (standard
examples include the dining philosopher problem in which case \emph{at least
one} philosopher must eat).
Secondly, $Q$-states target sets have been
considered in the literature of Petri nets:
e.g., in~\cite{AbdullaHM07}\footnote{Our definition seems
different from \cite{AbdullaHM07}, but
equivalent from standard embedding of Vector Addition Systems with States to Petri
Nets.} the authors showed that qualitative reachability for
probabilistic Vector Addition Systems with States with $Q$-states target
sets becomes decidable whereas the same problem is undecidable with
upward-closed target sets.

\begin{theorem} \label{thm-qual-reach-ptime}
Let $\S = (\Gamma,\mathord{\btran{}},\Prob)$ be a pBPP.
Let $Q \subseteq \Gamma$ represent an upward-closed set $F \subseteq \N^\Gamma$.
Let $\alpha_0 \in \N^\Gamma$ and $k \ge |\Gamma|$.
\begin{enumerate}
\item[(a)]
The coverability problem with $Q$-states target sets is solvable in polynomial
time; i.e., we can decide in polynomial time whether $\P(\alpha_0 \models \Diamond F) = 1$ holds.
\item[(b)]
We have: \vspace{-5mm}
\begin{align*}
                     \quad & \P(\alpha_0 \models \Diamond F) = 1 \\
 \Longleftrightarrow \quad & \P_\sigma(\alpha_0 \models \Diamond F) = 1 \text{ holds for some scheduler~$\sigma$} \\
 \Longleftrightarrow \quad & \P_\sigma(\alpha_0 \models \Diamond F) = 1 \text{ holds for all $k$-fair schedulers~$\sigma$.}
\end{align*}
As a consequence of part~(a), the existential and the $k$-fair universal problem are decidable in polynomial time.
\end{enumerate}
\end{theorem}
\begin{proof}
Denote by $Q' \subseteq \Gamma$ the set of types $X \in \Gamma$ such that
 there are $\ell \in \N$,
 a path $\alpha_0, \ldots, \alpha_\ell$ in the Markov chain $\DTMC(\S)$,
 and a type $Y \in Q$
  such that $\alpha_0 = X$ and $\beta = \alpha_\ell \ge Y$.
Clearly we have $Q \subseteq Q' \subseteq \Gamma$, and $Q'$ can be computed in polynomial time.

In the following, view~$\S$ as a context-free grammar with empty terminal set
 (ignore the probabilities, and put the symbols on the right-hand sides in an arbitrary order).
Remove from~$\S$ all rules of the form: (i) $X \btran{} \alpha$ where
 $X \in Q$ or $\alpha(Y) \ge 1$ for some $Y \in Q$,
and (ii) $X \btran{} \alpha$ where
$X \in \Gamma \setminus Q'$. Furthermore, add rules $X \btran{} \varepsilon$ where
$X \in \Gamma \setminus Q'$.
Check (in polynomial time) whether in the grammar the empty word~$\varepsilon$ is produced by~$\alpha_0$.

We have that $\varepsilon$ is produced by~$\alpha_0$ if and only if $\P(\alpha_0 \models \Diamond F) < 1$.
This follows from Proposition~\ref{prop-qual-reach-equal},
 as the complement of~$\widetilde{F}$ is the $Q'$-states set.
Hence part~(a) of the theorem follows.

For part~(b), let $\P(\alpha_0 \models \Diamond F) < 1$.
By part~(a) we have that $\varepsilon$ is produced by~$\alpha_0$.
Then for all schedulers~$\sigma$ we have $\P_\sigma(\alpha_0 \models \Diamond F) < 1$.
Trivially, as a special case, this holds for some $k$-fair scheduler.
(Note that $k$-fair schedulers exist, as $k \ge |\Gamma|$.)

Conversely, let $\P(\alpha_0 \models \Diamond F) = 1$.
By part~(a) we have that $\varepsilon$ is not produced by~$\alpha_0$.
Then, no matter what the scheduler does, the set~$F$ remains reachable.
So all $k$-fair schedulers will, with probability~$1$, hit~$F$ eventually.
\qed
\end{proof}

\section{Semilinear Target Sets} \label{sec-semilinear}

In this section, we prove that the qualitative reachability
problems that we considered in the previous sections become undecidable
when we extend upward-closed to semilinear target sets.

\newcommand{\stmtthmundecidability}{
Let $\S = (\Gamma,\mathord{\btran{}},\Prob)$ be a pBPP.
Let $F \subseteq \N^\Gamma$ be a semilinear set.
Let $\alpha_0 \in \N^\Gamma$.
The following problems are undecidable:
\begin{itemize}
\item[(a)]
Does $\P(\alpha_0 \models \Diamond F) = 1$ hold?
\item[(b)]
Does $\P_\sigma(\alpha_0 \models \Diamond F) = 1$ hold for all $7$-fair schedulers~$\sigma$?
\item[(c)]
Does $\P_\sigma(\alpha_0 \models \Diamond F) = 1$ hold for some scheduler~$\sigma$?
\end{itemize}
}
\begin{theorem} \label{thm-undecidability}
\stmtthmundecidability
\end{theorem}
The proofs are reductions from the control-state-reachability problem for $2$-counter machines,
 \iftechrep{see Appendix~\ref{app-undecidability}}{see~\cite{BKL14-fossacs-report}}.

\section{Conclusions and Future Work}
\label{sec:conclusion}

In this paper we have studied fundamental qualitative coverability and other reachability properties for pBPPs.
For the Markov-chain model, the coverability problem for pBPPs is decidable, which is in contrast to general pVASSs.
We have also shown a nonelementary lower complexity bound.
For the MDP model, we have proved decidability of the existential and the $k$-fair version of the universal coverability problem.
The decision algorithms for the MDP model are not (known to be) elementary,
 as they rely on Petri-net reachability and a Karp-Miller-style construction, respectively.
It is an open question whether there exist elementary algorithms.
Another open question is whether the universal MDP problem without any fairness constraints is decidable.

We have given examples of problems where the answer depends on the exact probabilities in the pBPP.
This is also true for the reachability problem for finite sets:
Given a pBPP and $\alpha_0 \in \N^\Gamma$ and a \emph{finite} set $F \subseteq \N^\Gamma$,
 the reachability problem for finite sets asks whether we have $\P(\alpha_0 \models \Diamond F) = 1$ in the Markov chain~$\DTMC(\S)$.
Similarly as in Example~\ref{ex-gambler-ruin} the answer may depend on the exact probabilities:
 consider the pBPP with $X \btran{p} X X$ and $X \btran{1-p} \varepsilon$, and let $\alpha_0 = X X$ and $F = \{ X \}$.
Then we have $\P(\alpha_0 \models \Diamond F) = 1$ if and only if $p \le 1/2$.
The same is true in both the existential and the universal MDP version of this problem.
Decidability of all these problems is open, but clearly decision algorithms would have to use techniques
 that are very different from ours, such as analyses of multidimensional random walks.

On a more conceptual level we remark that the problems studied in this paper are qualitative in two senses:
 (a) we ask whether certain events happen with probability~$1$ (rather than $>0.5$ etc.); and
 (b) the exact probabilities in the rules of the given pBPP do not matter.
Even if the system is nondeterministic and not probabilistic,
 properties (a) and~(b) allow for an interpretation of our results in terms of nondeterministic BPPs,
 where the nondeterminism is constrained by the laws of probability,
 thus imposing a special but natural kind of fairness.
It would be interesting to explore this kind of ``weak'' notion of probability for other (infinite-state) systems.

\subsubsection*{Acknowledgment}
Anthony W. Lin did this work when he was an EPSRC research fellow at Oxford
University supported by grant number EP/H026878/1.

    \bibliographystyle{plain}
    \bibliography{references}
\iftechrep{
    \newpage
    \appendix
    \section{Proofs of Section~\ref{sec-markov-chain}} \label{app-markov-chain}

\begin{qproposition}{\ref{prop-globally-coarse}}
 \stmtpropgloballycoarse
\end{qproposition}
\begin{proof}
The statement about~$\Mt(\S)$ follows from~\cite[Theorem~4.3]{AbdullaHM07}.
For the statement about~$\Mp(\S)$ define
\[
 \Min = \{\alpha \in \N^\Gamma \mid \Pp(\alpha \models \Diamond F) > 0 \text{ and for all } \alpha' < \alpha : \Pp(\alpha' \models \Diamond F) = 0\}\,.
\]
Note that $\Min$ is finite (this follows from Dickson's lemma).
Define $c := \min_{\alpha \in \Min} \Pp(\alpha \models \Diamond F)$.
Let $\gamma \in \N^\Gamma$ with $\Pp(\gamma \models \Diamond F) > 0$.
We prove the proposition by showing $\Pp(\gamma \models \Diamond F) \ge c$.

Take $\Gamma_\bullet := \Gamma \cup \Gamma'$ where $\Gamma' = \{X' \mid X \in \Gamma\}$ is a copy of~$\Gamma$.
Similarly, we clone the rules so that we get $\mathord{\btran{}_\bullet} \subseteq (\Gamma \times \N^\Gamma) \cup (\Gamma' \times \N^{\Gamma'})$
 and define $\Prob_\bullet$ in the obvious way.
Let $\S_\bullet = (\Gamma_\bullet,\mathord{\btran{}_\bullet},\Prob_\bullet)$.
Let $\P_\bullet$ denote the probability measure of $\Mp(\S_\bullet)$.

Partition~$\gamma$ in $\gamma = \alpha + \beta$ where $\alpha \in \Min$,
 and let $\beta' \in \N^{\Gamma'}$ be a clone of~$\beta$.
We have:
\begin{align*}
 \Pp(\gamma \models \Diamond F)
 &  = \P_\bullet(\gamma \models \Diamond F)              && \text{definition of~$\P_\bullet$}\\
 & \ge \P_\bullet(\alpha + \beta' \models \Diamond F)    && \text{as $F \subseteq \N^\Gamma$}\\
 &  = \P_\bullet(\alpha \models \Diamond F)              && \text{see below}\\
 &  =  \Pp(\alpha \models \Diamond F)                    && \text{definition of~$\P_\bullet$}\\
 & \ge c                                                 && \text{by definition of~$c$}
\end{align*}
To show the equality $\P_\bullet(\alpha + \beta' \models \Diamond F) = \P_\bullet(\alpha \models \Diamond F)$
 we show that as long as there are $\Gamma$-processes originating from $\alpha + \beta'$
  (i.e., processes originating from~$\alpha$), they are eventually scheduled with probability~$1$.
In fact, let $\varepsilon \ne \alpha \in \N^\Gamma$ and $\beta' \in \N^{\Gamma'}$ be arbitrary.
Let $z = \max_{X \btran{} \delta} \Wp{\delta}$ be a bound on the number of processes that can be created per step.
Let $a := \Wp{\alpha}$ and $b := \Wp{\beta'}$.
It suffices to show that the probability that only $\Gamma'$-processes are scheduled is~$0$.
This probability is at most
\[
 \frac{b}{a+b} \cdot \frac{b+z}{a+b+z} \cdot \frac{b+2z}{a+b+2z} \cdot \ldots
\]
Recall that for a sequence $(a_i)_{i \in \N}$ with $a_i \in (0,1)$ we have $\prod_{i \in \N} (1-a_i) = 0$ if and only if
 the series $\sum_{i \in \N} a_i$ diverges.
It follows that the infinite product above is~$0$.
\qed
\end{proof}


\begin{qlemma}{\ref{lem-qual-reach-decid}}
 \stmtlemqualreachdecid
\end{qlemma}
\begin{proof}
For all $\ell \in \{0, \ldots, k-1\}$ let $X_\ell \btran{} \delta_\ell$ be a rule with
 $\alpha_\ell(X_\ell) \ge 1$ and $\alpha_\ell - X_\ell + \delta_\ell = \alpha_{\ell+1}$.
Assume for a contradiction that $i < j$ with $\alpha_i \le \alpha_j$.
For all $\ell \in \{j, \ldots, k-1\}$ define $\alpha'_\ell \in \N^\Gamma$ and $\beta_\ell \in \N^\Gamma$ so that $\alpha_\ell = \alpha'_\ell + \beta_\ell$
 and $\alpha'_j = \alpha_i$ and
\begin{itemize}
 \item $\alpha'_\ell(X_\ell) \ge 1$ and $\alpha'_\ell - X_\ell + \delta_\ell = \alpha'_{\ell+1}$ \quad or
 \item $\beta_\ell(X_\ell) \ge 1$ and $\beta_\ell - X_\ell + \delta_\ell = \beta_{\ell+1}$.
\end{itemize}
As $\bar{F}$ 
 is downward-closed, we have $\alpha'_\ell \in \bar{F}$ for all $\ell \in \{j, \ldots, k\}$. 
It follows that
 \[
  \alpha_0 \to \alpha_1 \to \ldots \alpha_{i-1} \to \alpha'_j \to \alpha'_{j+1} \to \ldots \to \alpha'_k \le \gamma
 \]
is, after removing repetitions, a path via $\bar{F}$-states.
As $i < j$, the path is shorter than the path $\alpha_0 \to \alpha_1 \to \ldots \to \alpha_k$, so we have obtained the desired contradiction.
\qed
\end{proof}

    \newcommand{\MM}{\mathcal{M}}
\newcommand{\NN}{\mathcal{N}}
\newcommand{\ba}[1]{\begin{array}{#1}}
\newcommand{\ea}{\end{array}}
\newcommand{\caps}[1]{\text{\textsc{\scriptsize #1}}}

\section{Proof of the Lower Complexity Bound}
\label{sec-complexity}

\begin{qtheorem}{\ref{thm-nonelementary}}
 \stmtthmnonelementary
\end{qtheorem}

\bigskip
\noindent
We claim that there exists a nonelementary function $f$ such that given a
Turing machine $\mathcal{M}$ running in space $f(k)$, we can build a pBPP
$\S = (\Gamma, \mathord{\btran{}}, \Prob)$ of size $k$, and
an upward closed set $F \subseteq \N^\Gamma$ such that
$\P(s_0 \models \bar{F} ~\U~ \widetilde{F}) = 0$ if and only if $\mathcal{M}$
doesn't terminate.

\smallskip

First, let us mention that we can change $\widetilde{F}$ by any downward closed
subset of $\bar{F}$. Assume for example that we wish to reach $\bar{G}$ for
some upward closed set $G$. This can be done by adding new processes $T$
and $T_2$, replacing $F$ by $F_2 = T^+F \cup T_2^+G$ and adding the
following transitions:

$$
T \btran{} T_2 \qquad T_2 \btran{} \eps \qquad
\forall X \in \Gamma.~ T \btran{} TX
$$

Then, if there is a way to reach $s \in \bar{G}$ by staying into $\bar{F}$ in
the original net, this means you can reach $Ts$ by staying into $T\bar{F}$ in
the modified net. Then, you can go into $T_2s$, which is not in $T_2 G$, and
then to $s$, that doesn't contain either $T$ or $T_2$, which means that is
in $\widetilde{F_2}$ given that nothing can spawn these processes anymore.
Reciprocally, if you can find a way to reach $\widetilde{F_2}$ in the modified
net, this means you have been able to successfully consume $T$ (otherwise you
could spawn any process), which means that there was a path in $T\bar{F}$
reaching a configuration $T_2s$ with $s \in \bar{G}$. This path could include
spurious process spawns from the process $T$, but by monotony, we can remove
these spawns, and get a path of the original net in $\bar{F}$ reaching a
configuration $s \in \bar{G}$.

\smallskip

We ignore the probability part, as it doesn't matter for our reachability
question and we call a BPP given with a set $\bar{F}$ a \emph{constrained BPP},
in which we consider only paths that stay inside $\bar{F}$.

\smallskip

Now, we build a constrained BPP that can simulate a Minsky machine.
In order to do that, let us remark that, instead of defining explicitly the
set $F$, one can list contraints that will define the allowed set $F$ as the
union of the sets $F$ implied by each constraint. The two basic types
of constraints that we will use are:

\begin{itemize}
\item Processes $X$ and $Y$ are incompatible (i.e. if $X$ is present,
$Y$ cannot be, and vice-versa). This is associated to $XY \uc$.
\item Process $X$ is unique (you can't have two copies of it). This is
associated to $XX \uc$.
\end{itemize}

We will allow ourselves to use more complex restrictions, that can be encoded
in this system by adding extra processes. These are:

\begin{itemize}
\item Process $X$ prevents rule $Y \btran{} u$ to be fired. This is done by
adding a dummy process $T$ and the following rules:

$$
Y \btran{} Tu \qquad T \btran{} \eps \qquad F = TX \uc
$$

\item Given a subpart of a BPP $\NN$ (that is, a set of process types, and a 
set of rules refering only to these process types), with a downward closed set 
of initial configurations $I$ and a downward closed set of final configurations
$F$, the subpart $\NN$ has an atomic behaviour: rules that don't belong to
$\NN$ can't be used if $\NN$ is not in $I$ or $F$. This is done by replacing 
each rule $Y \btran{} u$ outside $\NN$ by $Y \btran{} uT$ and 
$T \btran{} \eps$ with the added
constraint $(\overline{I \cup F})T \uc$.
\end{itemize}

We will use Petri Net-style depictions of BPPs, where process types are called
"places" and represented as circles while processes are called tokens and are
represented as bullets in their associated circle. A transition turning a
process $X$ into processes $Y_1 ... Y_k$ is represented as an arrow linking
the place $X$ to the places $Y_1 ... Y_k$. Moreover, we represent unique
processes (places that can contain only one token) as squares instead of
circles.

\subsection{Consumers, Producers and Counters}

We look at three specific kind of constrained BPP:

\begin{itemize}
\item A $k$-producer is a constrained BPP $\NN$ with an initial configuration
$s_i$, one data place $X$ and a downward closet subset of final configuration
$F$ such that if $s_f \in F$ is reachable from $s_i$, then $s_f(X) = k$.
\item A $k$-consumer is a constrained BPP $\NN$ with an initial configuration
 $s_i$, one data place $X$ and a downward closet subset of final configuration
$F$ such that for every $p \in \N$, $\NN$ can reach $F$ from $s_i + X^p$ if
and only $p \leq k$.
\item A $k$-loop is a constrained BPP $\NN$ with three disjoint
upward-closed set of configurations $A$, $B$, $C$, an initial configuration
$s_i$ and a downward-closed set of final configurations $F$ such that for
every run that goes from $s_i$ to $F$, the net always stay in
$A \cup B \cup C$ and cycles through these three sets, in the
order $A, B, C$ exactly $k$ times.
\end{itemize}

Intuitively, a $k$-producer is a gadget that forces the appearance of at least
$k$ tokens. This allows to force an operation to run more than $k$ times (by 
running the producer, allowing the operation to make exactly one token 
disappear, then require that all tokens have disappeared). Symmetrically, a 
$k$-consumer is a gadget that is able to consume up to $k$ 
tokens. This allows to restrain an operation to run up to $k$ times (by making
it generate such a token, then requiring these tokens to have disappeared).
Finally, a $k$-loop is a gadget that has a controlled cyclic behavior which
occurs exactly $k$ times. By syncing it with another gadget, this will allow
to make an operation run exactly $k$ times (it is basically a combination of
a producer and a consumer).

\smallskip

In all the following lemmas, note that the number of constraints is polynomial
in the number of places.

\begin{lemma}
There exists a constant $\alpha$ such that given a $k$-loop with $n$
places, one can build a $2^k$-producer with $n + \alpha$ places.
\end{lemma}

\begin{proof}
This producer will be of the following form:

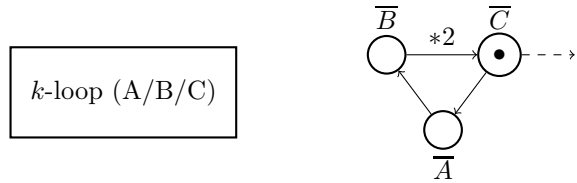
\begin{center}
\begin{tikzpicture}
\tikzstyle{place} = [circle,thick,draw,minimum size=5mm];
\tikzstyle{bplace} = [thick,draw,minimum size=5mm];
\tikzstyle{subnet} = [rectangle,thick,draw,minimum width=3cm,
              minimum height=1.2cm];

\node[subnet] (counter) at (0,0) {$k$-loop (A/B/C)};
\node[place] (B2) at (3.5,0.5) {};
\node[above of=B2,node distance=0.5cm] {$\overline{B}$};
\node[place] (C2) at (5,0.5) {$\bullet$};
\node[above of=C2,node distance=0.5cm] {$\overline{C}$};
\node[place] (A2) at (4.25,-0.5) {};
\node[below of=A2,node distance=0.5cm] {$\overline{A}$};
\draw[->,dashed] (C2) -- (6,0.5);
\draw[->] (B2) -- (C2) node[pos=0.5,above] {$*2$};
\draw[->] (C2) -- (A2);
\draw[->] (A2) -- (B2);
\end{tikzpicture}
\end{center}

\underline{\textit{Definition of the Producer and Constraints:}}

\begin{itemize}
\item The place labelled by $\overline{C}$ is the final place that will contain
the required number of tokens in the final configuration.
\item The net is in its final configuration if the loop is in its final
configuration and
there is no more tokens in places labelled by $\overline{A}$, $\overline{B}$.
\item Places labelled by $\overline{A}$ (resp. $\overline{B}$, $\overline{C}$)
are incompatible with the loop being in configurations inside $A$
(resp. $B$, $C$).
\end{itemize}

In order to put the net in its final configuration, the loop must perform
$k$ cycles on $A$-$B$-$C$. This means that tokens in the places $\overline{A}$,
$\overline{B}$ and $\overline{C}$ must simultaneously move. Moreover, every
time the loop performs a cycle, the number of tokens is doubled. This means
that at the end, $2^k$ tokens will be in the final place $\overline{C}$.

\qed

\end{proof}

\begin{lemma}
There exists a constant $\alpha$ such that given a $k$-loop with $n$
places, one can build a $2^k$-consumer with $n + \alpha$ places.
\end{lemma}

\begin{proof}
This consumer will be of the following form:

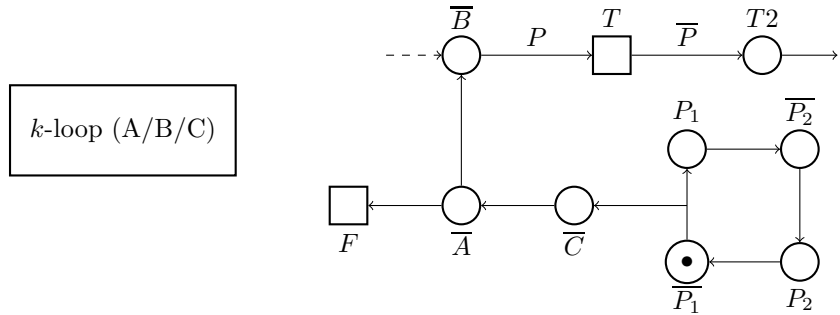
\begin{center}
\begin{tikzpicture}
\tikzstyle{place} = [circle,thick,draw,minimum size=5mm];
\tikzstyle{bplace} = [thick,draw,minimum size=5mm];
\tikzstyle{subnet} = [rectangle,thick,draw,minimum width=3cm,
              minimum height=1.2cm];

\node[subnet] (counter) at (0,0) {$k$-loop (A/B/C)};
\node[place] (B2I) at (4.5, 1) {};
\node[above of=B2I,node distance=0.5cm] {$\overline{B}$};
\node[bplace] (P) at (6.5, 1) {};
\node[above of=P,node distance=0.5cm] {$T$};
\node[place] (P2) at (8.5, 1) {};
\node[above of=P2,node distance=0.5cm] {$T2$};
\node[place] (A2) at (4.5, -1) {};
\node[below of=A2,node distance=0.5cm] {$\overline{A}$};
\node[place] (C2) at (6, -1) {};
\node[below of=C2,node distance=0.5cm] {$\overline{C}$};
\node[bplace] (F) at (3, -1) {};
\node[below of=F,node distance=0.5cm] {$F$};
\node[place] (S1) at (7.5, -0.25) {};
\node[above of=S1,node distance=0.5cm] {$P_1$};
\node[place] (S2) at (9, -0.25) {};
\node[above of=S2,node distance=0.5cm] {$\overline{P_2}$};
\node[place] (S3) at (9, -1.75) {};
\node[below of=S3,node distance=0.5cm] {$P_2$};
\node[place] (S4) at (7.5, -1.75) {$\bullet$};
\node[below of=S4,node distance=0.5cm] {$\overline{P_1}$};

\draw[->,dashed] (3.5, 1) -- (B2I);
\draw[->] (B2I) -- (P) node[pos=0.5,above] {$P$};
\draw[->] (P) -- (P2) node[pos=0.5,above] {$\overline{P}$};
\draw[->] (P2) -- (9.5, 1);
\draw[->] (S1) -- (S2);
\draw[->] (S2) -- (S3);
\draw[->] (S3) -- (S4);
\draw[->] (S4) -- (S1);
\draw[->] (7.5, -1) -- (C2);
\draw[->] (C2) -- (A2);
\draw[->] (A2) -- (F);
\draw[->] (A2) -- (B2I);
\end{tikzpicture}
\end{center}

\underline{\textit{Definition of the Consumer and Constraints:}}
\begin{itemize}
\item The place labelled by $\overline{B}$ is the initial place, that will
contain the initial number of tokens to consume.
\item The net is in its final configuration if the loop is in its final
configuration, and the only tokens in the remainder of the net are in the
places labelled by $\overline{P_1}$ and $F$.
\item Places labelled by $\overline{A}$ (resp. $\overline{B}$, $\overline{C}$)
are incompatible with the loop being in configurations inside $A$
(resp. $B$, $C$).
\item The transitions labelled by $P$ (resp. $\overline{P}$) are incompatible
with the presence of tokens in the places labelled by $\overline{P_1}$ and
$\overline{P_2}$ (resp. $P_1$ and $P_2$).
\item The place labelled $\overline{A}$ is incompatible with the presence of
a token in the places labelled by $P_1$, $P_2$, $\overline{P_2}$ and $T$.
\end{itemize}

Let us assume our initial configuration has $n$ tokens in the place
$\overline{B}$. We
are looking to a run that empties this place. In order to do that, let us look
at what happen in one step of the loop. When the loop in in configuration $A$,
the token in the places $P_1$, $P_2$, $\overline{P_1}$ and $\overline{P_2}$ can
cycle around these places. For each such cycle, two tokens can be removed from
the place $\overline{B}$, and one token is created in place $\overline{C}$.
Once the place $\overline{B}$ is empty, the loop can move into configuration
 $B$. There, in order to be able to move all tokens into place $\overline{A}$,
we must have the place $T$ empty, and the cycling token into place
$\overline{P_1}$. Now, we can move all tokens in place $\overline{A}$, the
loop in configuration $C$, then all tokens in place $\overline{B}$ and the
loop back in configuration $A$. This means our net is back in its original
configuration,
except the loop has performed one cycle, and the number of tokens in
$\overline{B}$ has been halved (rounded up). This can be done as many times
as the loop allows it, with the final iteration moving one token into $F$
instead of back into $\overline{B}$. This allows to consume up to $2^k$
tokens, where $k$ is the number of iterations of the loop.

\qed

\end{proof}

\begin{lemma}
There exists a constant $\alpha$ such that given a $k$-consumer with $n$ places
and a $k$-producer with $n$ places, one can build a $k$-loop with
$2n+\alpha$ places.
\end{lemma}

\begin{proof}

This loop will be of the following form:

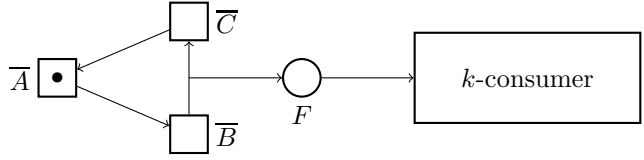
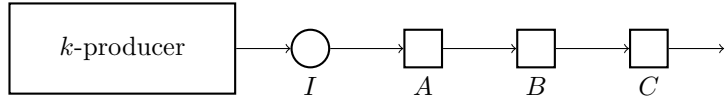
\begin{center}
\begin{tikzpicture}
\tikzstyle{place} = [circle,thick,draw,minimum size=5mm];
\tikzstyle{bplace} = [thick,draw,minimum size=5mm];
\tikzstyle{subnet} = [rectangle,thick,draw,minimum width=3cm,
              minimum height=1.2cm];

\node[subnet] (prod) at (0,0) {$k$-producer};
\node[place] (X) at (2.5,0) {};
\node[below of=X,node distance=0.5cm] {$I$};
\node[bplace] (A1) at (4,0) {};
\node[below of=A1,node distance=0.5cm] {$A$};
\node[bplace] (B1) at (5.5,0) {};
\node[below of=B1,node distance=0.5cm] {$B$};
\node[bplace] (C1) at (7,0) {};
\node[below of=C1,node distance=0.5cm] {$C$};
\draw[->] (prod) -- (X);
\draw[->] (X) -- (A1);
\draw[->] (A1) -- (B1);
\draw[->] (B1) -- (C1);
\draw[->] (C1) -- (8,0);

\node[bplace] (A2) at (1.25,2.75) {$\bullet$};
\node[left of=A2,node distance=0.5cm] {$\overline{A}$};
\node[bplace] (B2) at (3,2) {};
\node[right of=B2,node distance=0.5cm] {$\overline{B}$};
\node[bplace] (C2) at (3,3.5) {};
\node[right of=C2,node distance=0.5cm] {$\overline{C}$};
\node[place] (F) at (4.5,2.75) {};
\node[below of=F,node distance=0.5cm] {$F$};
\node[subnet] (consumer) at (7.5,2.75) {$k$-consumer};
\draw[->] (A2) -- (B2);
\draw[->] (B2) -- (C2);
\draw[->] (C2) -- (A2);
\draw[->] (3,2.75) -- (F);
\draw[->] (F) -- (consumer);
\end{tikzpicture}
\end{center}

\underline{\textit{Definition of the loop and Constraints:}}

\begin{itemize}
\item The cycles of the loop are associated to the token cycling on
the places $\overline{A}$, $\overline{B}$ and $\overline{C}$.
\item The loop is in its final configuration if the producer and the consumer
are in their final configuration, and the places $I$, $A$, $B$ and $C$ are
empty.
\item The places $A$, $B$ and $C$ are mutually exclusive.
\item The place $A$ (resp. $B$, $C$) are incompatible with the presence of
a token in the place $\overline{A}$ (resp. $\overline{B}$, $\overline{C}$).
\end{itemize}

Let us look at a run going to the final configuration. In order to do that, the
producer has created $k$ tokens into place $I$. This means that a token has
been through $A$, $B$ and $C$ at least $k$ times, which means that the token
has cycled through $\overline{A}$, $\overline{B}$ and $\overline{C}$ at least
$k$ times. Moreover, whenever such a cycle has been performed, one token has
been created into $F$. As the consumer can consume only up to $k$ tokens, it
means there was also at most $k$ cycles.
\qed
\end{proof}

\subsection{Simulating a bounded counter machine}

In this section, we simulate a counter machine whose counters are bounded
by a constrained BPP. Our construction is made of one scheduler (see figure
\ref{fig:scheduler}) and as many counters as the machine we want to simulate
(see figure \ref{fig:counter}).

\bigskip

During the execution, when the scheduler is entering \caps{STEP1}, the place
$C$ of each counter will contain its value, and $\overline{C}$ its complement.
During
steps 2 to 5, most of the tokens from $C$ and $\overline{C}$ will be transfered
respectively to $B$ and $\overline{B}$. The places \caps{INCREMENT},
\caps{DECREMENT} and
\caps{ZERO-TEST} of each counter are called the operationnal places, and
contain a
token when the counter is currently performing an operation.
Finally, for each transition of the machine we are simulating, we have a
transition between \caps{STEP1} and \caps{STEP2} that fills for each counter
the correct operationnal place.

\bigskip

In order to ensure our counters perform the operations requested, we have the
following constraints (encoded in our upward closed set, as before):
\begin{itemize}
\item \caps{INIT}: The producer can only run during init. He must have finished
running before entering \caps{STEP1}.
\item \caps{STEP1}: Tokens may be moved from $C$ and $\overline{C}$ to $B$ and
$\overline{B}$. The consumer can be reset.
\item \caps{STEP2}: In order to enter this step, token
repartition must match the token that is simultaneously appearing
in the operationnal place:
\begin{itemize}
\item \caps{INCREMENT}: At most one token in $\overline{C}$, and none in $C$.
\item \caps{DECREMENT}: At most one token in $C$, and none in $\overline{C}$.
\item \caps{ZERO-TEST}: No tokens in $C$ or $B$.
\end{itemize}
Tokens in $T$ can be consumed according to the capacity of the
consumer.
\item \caps{STEP3}: In order to enter this step, $T$ must be empty.
Tokens may be moved freely between $C$ and $\overline{C}$.
\item \caps{STEP4}: In
order to enter this step, token repartition must match the token that is in
the operationnal place:
\begin{itemize}
\item \caps{INCREMENT}: At most one token in $C$, and none in $\overline{C}$.
\item \caps{DECREMENT}: At most one token in $\overline{C}$, and none in $C$.
\item \caps{ZERO-TEST}: No tokens in $C$ or $B$.
\end{itemize}
\item \caps{STEP5}: Tokens in the operational places can be deleted.
Tokens in $B$
and $\overline{B}$ can go back to $C$ and $\overline{C}$.
\item \caps{STEP6}: In order to enter this step, there must be no tokens in $B$,
$\overline{B}$ or in operational places.
\end{itemize}

\begin{figure}
\begin{center}
\begin{tikzpicture}
\tikzstyle{bplace} = [thick,draw,minimum size=5mm];

\node[bplace] (INIT) at (0, 0) {$\bullet$};
\node (lINIT) at (0, 0.5) {\tiny init};

\node[bplace] (STEP1) at (1.5, 0) {};
\node (l1) at (1.5, 0.5) {\tiny step1};
\draw[->] (INIT) -- (STEP1);

\node[bplace] (STEP2) at (3, 0) {};
\node (l2) at (3, 0.5) {\tiny step2};
\draw[->] (STEP1) -- (STEP2);

\node (label) at (2.25, 1.5)
{\tiny $\ba{c} \text{towards the operational} \\
\text{places of the counters} \ea$};
\draw[->] (2.25, 0) -- (label);

\node[bplace] (STEP3) at (4.5, 0) {};
\node (l3) at (4.5, 0.5) {\tiny step3};
\draw[->] (STEP2) -- (STEP3);

\node[bplace] (STEP4) at (6, 0) {};
\node (l1) at (6, 0.5) {\tiny step4};
\draw[->] (STEP3) -- (STEP4);

\node[bplace] (STEP5) at (7.5, 0) {};
\node (l1) at (7.5, 0.5) {\tiny step5};
\draw[->] (STEP4) -- (STEP5);

\node[bplace] (STEP6) at (9, 0) {};
\node (l1) at (9, 0.5) {\tiny step6};
\draw[->] (STEP5) -- (STEP6);

\draw (STEP6) -- (9, -1.5);
\draw (9, -1.5) -- (1.5, -1.5);
\draw[->] (1.5, -1.5) -- (STEP1);

\end{tikzpicture}
\end{center}
\caption{The scheduler of our simulated machine}
\label{fig:scheduler}
\end{figure}
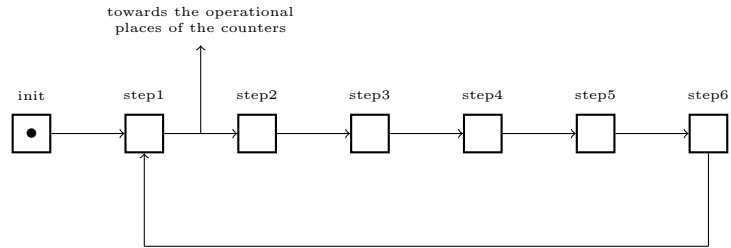

\begin{figure}
\begin{center}
\begin{tikzpicture}
\tikzstyle{place} = [circle,thick,draw,minimum size=5mm];
\tikzstyle{bplace} = [thick,draw,minimum size=5mm];
\tikzstyle{subnet} = [rectangle,thick,draw,minimum width=3cm,
              minimum height=1.2cm];

\node[place] (C1) at (0,0) {};
\node[above of=C1,node distance=0.5cm] {$C$};
\node[place] (C2) at (0, -3) {};
\node[below of=C2,node distance=0.5cm] {$\overline{C}$};
\node[bplace] (OPPLUS) at (2, -4.5) {};
\node[bplace] (OPMOINS) at (2, -6) {};
\node[bplace] (OPTEST) at (2, -7.5) {};
\node[place] (B1) at (4, 0) {};
\node[above of=B1,node distance=0.5cm] {$B$};
\node[place] (B2) at (4, -3) {};
\node[below of=B2,node distance=0.5cm] {$\overline{B}$};
\node[place] (CHECK) at (2, -1.5) {};
\node[left of=CHECK,node distance=0.5cm] {$T$};

\node[subnet] (CONSUMER) at (5.5, -1.5) {$k-1$-consumer};
\node[subnet] (PRODUCER) at (-3, -3) {$k$-producer};

\draw[->] (C1) edge[bend left=20] (C2);
\node (edgel1) at (-0.7,-1.5) {\tiny step3};
\draw[->] (C2) edge[bend left=20] (C1);
\node (edgel2) at (0.7, -1.5) {\tiny step3};
\draw[->] (C1) edge (B1);
\draw[->] (B1) edge[bend right=25] (C1);
\node (edgel3) at (2, 0.7) {\tiny step5};
\draw[->] (2, 0) -- (CHECK) node[left,pos=0.2] {\tiny step1};
\draw[->] (C2) edge (B2);
\draw[->] (B2) edge[bend left=25] (C2);
\node (edgel4) at (2, -3.7) {\tiny step5};
\draw[->] (2, -3) -- (CHECK) node[left,pos=0.2] {\tiny step1};
\draw[->] (PRODUCER) -- (C2) node[above,pos=0.5] {\tiny init};
\draw[->] (CHECK) -- (CONSUMER) node[above,pos=0.5] {\tiny step2};

\node (nodel1) at (2, -5) {\tiny INCREMENT};
\node (nodel2) at (2, -6.5) {\tiny DECREMENT};
\node (nodel3) at (2, -8) {\tiny ZERO-TEST};

\node (schedl1) at (-1, -4.5)
{\tiny $\ba{c} \text{from scheduler} \\ \text{during step1 to step2} \ea$};
\draw[->] (schedl1) -- (OPPLUS);
\node (schedl1b) at (3.8, -4.5) {};
\draw[->] (OPPLUS) -- (schedl1b) node[above,pos=0.5] {\tiny step5};

\node (schedl2) at (-1, -6)
{\tiny $\ba{c} \text{from scheduler} \\ \text{during step1 to step2} \ea$};
\draw[->] (schedl2) -- (OPMOINS);
\node (schedl2b) at (3.8, -6) {};
\draw[->] (OPMOINS) -- (schedl2b) node[above,pos=0.5] {\tiny step5};

\node (schedl3) at (-1, -7.5)
{\tiny $\ba{c} \text{from scheduler} \\ \text{during step1 to step2} \ea$};
\draw[->] (schedl3) -- (OPTEST);
\node (schedl3b) at (3.8, -7.5) {};
\draw[->] (OPTEST) -- (schedl3b) node[above,pos=0.5] {\tiny step5};
\end{tikzpicture}
\end{center}
\caption{A counter of our simulated machine}
\label{fig:counter}
\end{figure}
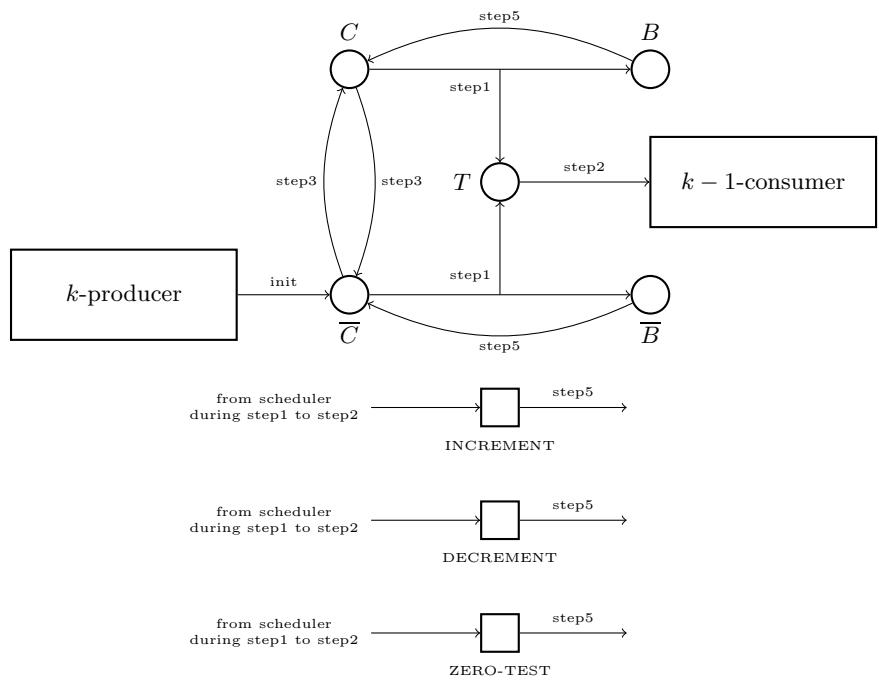

We claim that performing steps 2 to 5 moves the counter according to the
instruction given by the token in \caps{INCREMENT}, \caps{DECREMENT} or
\caps{ZERO-TEST}.

\smallskip

The case of \caps{ZERO-TEST} is simple: firing no transitions in the counter
allows to move from step 2 to step 5 unimpended, and the restrictions on
step 2 and 4 make sure that there must be no tokens in $C$ before or after
these operations. This means that the counter value must be zero and stays as
such.

\smallskip

For the case of \caps{INCREMENT}, the restrictions on entering step 2 means
that all tokens must be moved from $C$ to $B$, and all (except possibly one)
must be moved from $\overline{C}$ to $\overline{B}$. However, moving a token
creates a token in $T$, which means that in order to fulfill the restrictions
of step 3, at most $k-1$ tokens can be moved. Because the total number of
tokens in $C$, $B$, $\overline{C}$ and $\overline{B}$ is always $k$, and there
was no tokens in $B$ or $\overline{B}$ before this transfer, this means that
exactly $k-1$ tokens have been moved and that exactly one token remains in
$\overline{C}$. Requirements of step 4 means that this token must be moved
from $\overline{C}$ to $C$. Finally, the requirements of step 6 means that
after step 5, all tokens in $B$ and $\overline{B}$ have moved back to their
original places, which means that the number of tokens in $C$ has increased
by 1.

\smallskip

The case of \caps{DECREMENT} is symmetric.

    \section{Proofs of Section~\ref{sec-mdp}} \label{app-mdp}

\begin{qtheorem}{\ref{thm-mdp-existential}}
 \stmtthmmdpexistential
\end{qtheorem}

\begin{proof}
\newcommand{\suc}{\mathit{succ}}%
\newcommand{\Sat}{\mathit{Sat}}%
\newcommand{\Stable}{\mathit{Stable}}%
\newcommand{\Unstable}{\mathit{Unstable}}%
\newcommand{\down}[1]{\lfloor #1 \rfloor}%
 In this proof we say the scheduler \emph{wins} if the system reaches the upward-closed set~$F$.
 We also say that a configuration~$\alpha$ is \emph{winning} with probability~$1$
  if there exists a scheduler~$\sigma$ such that $\P_\sigma(\alpha \models \Diamond F) = 1$.
 The question is whether the initial configuration~$\alpha_0$ is winning with probability~$1$.

 The proof idea is to construct a finite-state MDP whose state space~$Q$ is a finite subset of $\N^\Gamma$.
 In the finite-state MDP the scheduler wins if a state from $Q \cap F$ is reached.
 The question whether the scheduler can win with probability~$1$ is decidable in polynomial time for finite-state MDPs.
 Moreover, for reachability with probability~$1$ deterministic and memoryless schedulers suffice and can be computed efficiently,
  see e.g.~\cite{ChatterjeeH11} or the references therein.
 The actions in the finite-state MDP are as in the infinite-state MDP (each type corresponds to an action),
  but we need to redirect transitions that would leave the finite state space.

 Define the directed graph~$G$ with vertex set~$\Gamma$ and edges $(X,Y)$ whenever
  there is a rule $X \btran{} \beta$ with $\beta(Y) \ge 1$.
 We say that $Y$ is a \emph{successor} of~$X$ if $(X,Y)$ is in the reflexive and transitive closure of the edge relation.
 For $X \in \Gamma$, we write $\suc(X)$ for the set of successors of~$X$.
 Note that $X \in \suc(X)$ for all $X \in \Gamma$.
 Let $\phi_1, \ldots, \phi_m$ denote the minimal elements of~$F$.
 Define $K := \max \{ \phi_i(X) \mid 1 \le i \le m, \ X \in \Gamma \}$.
 Let $\alpha \in \N^\Gamma$.
 Define the set of \emph{saturated} types by $\Sat(\alpha) := \{X \in \Gamma \mid \forall Y \in \suc(X): \alpha(Y) \ge K\}$.
 Define $\down{\alpha} \in \N^\Gamma$ by $\down{\alpha}(X) = K$ for $X \in \Sat(\alpha)$ and $\down{\alpha}(X) = \alpha(X)$ for $X \not\in \Sat(\alpha)$.
 Note that $\Sat(\alpha) = \Sat(\down{\alpha})$.
 We make the following observation.
 \begin{itemize}
  \item[(1)]
   Let $\alpha \in \N^\Gamma$ be winning with probability~$1$.
   Then the scheduler can win with probability~$1$ by \emph{never} scheduling a type $X \in \Sat(\alpha)$.
   Moreover, $\down{\alpha}$ is winning with probability~$1$.
 \end{itemize}
 For a configuration $\alpha \in \N^\Gamma$, define $\Stable(\alpha) := \{X \in \Gamma \mid \alpha(X) \le K\} \cup \Sat(\alpha)$
  and $\Unstable(\alpha) := \N^\Gamma \setminus \Stable(\alpha) = \{X \in \Gamma \mid \alpha(X) > K \text{ and } \exists Y \in \suc(X): \alpha(Y) < K\}$.
 We call $\alpha$ \emph{stable} resp.\ \emph{unstable} if $\Unstable(\alpha) = \emptyset$ resp.\ $\Unstable(\alpha) \ne \emptyset$.

 We define a finite-state MDP so that the scheduler can win with probability~$1$ if and only if
  it can win in the original MDP with probability~$1$.
 (In fact, we even show that the optimal winning probability stays the same.)
 The set of states of the finite-state MDP is
  \[
   Q := \{\down{\alpha} \mid \alpha \in \N^\Gamma \text{ is stable}\} \subseteq \{0, \ldots, K\}^\Gamma\,.
  \]
 Note that $|Q| \le (K+1)^{|\Gamma|}$.
 The target states are those in~$F$.
 The actions are as in the original infinite-state MDP, i.e., if $\alpha(X) \ge 1$, then scheduling~$X$ is a possible action in~$\alpha$.
 As in the infinite-state MDP, there is a special action~$\bot$ enabled only in the empty configuration~$\varepsilon$
  (which is losing for the scheduler except in trivial instances).
 If an action can lead to a state not in~$Q$, we need to redirect those transitions to states in~$Q$ as we describe in the following.
 If a transition leads to a stable configuration~$\alpha$ outside of~$Q$, then the transition is redirected to $\down{\alpha} \in Q$, following Observation~(1).
 (Also by Observation~(1), the actions corresponding to types $X \in \Sat(\alpha) = \Sat(\down{\alpha})$
   could be disabled without disadvantaging the scheduler.)
 If a transition leads to an unstable configuration~$\alpha$, then this transition is redirected to a probability distribution~$p_\alpha$ on~$Q$
  so that for each $q \in Q$ we have that $p_\alpha(q)$ is the probability that in the original infinite-state MDP a configuration~$\beta \in \N^\Gamma$
   with $\down{\beta} = q$ is the first stable configuration reached when following a particular class of optimal strategies which we describe in the following.

 The strategy class relies on the fact that in a configuration~$\alpha \in \N^\Gamma$ with $\alpha(X) > K$ for some $X \in \Gamma$,
  the scheduler does not suffer a disadvantage by scheduling~$X$:
  Indeed, by scheduling $X$, only the $X$-component of the configuration can decrease, and if it decreases, it decreases by at most~$1$;
   so we have $\alpha'(X) \ge K$ also for the successor configuration~$\alpha'$.
 As those $X$-processes in excess of~$K$ can only be useful for producing types that are successors of~$X$,
  one can schedule them freely and at any time.
 We call a strategy \emph{cautious} if it behaves in the following way while the current configuration~$\alpha \in \N^\Gamma$ is unstable:
 \begin{quote}
  Let $X_1, \ldots, X_k$ with $k \le |\Gamma|$ be the shortest path (where ties are resolved in an arbitrary but deterministic way)
   in the graph~$G$ from the beginning of this proof
   such that $\alpha(X_1) > K$ and $\alpha(X_k) < K$ and $\alpha(X_i) = K$ for $2 \le i \le k-1$.
  Schedule~$X_1$.
 \end{quote}
 We claim that, with probability~$1$, a stable configuration will eventually be reached if a cautious strategy is followed.
 To see that, consider an unstable configuration~$\alpha$ and let $X_1 \in \Gamma$ be scheduled, i.e.,
  $X_1, \ldots, X_k$ with $k \le |\Gamma|$ is the shortest path in the graph~$G$
  such that $\alpha(X_1) > K$, and $\alpha(X_k) < K$ and $\alpha(X_i) = K$ for $2 \le i \le k-1$.
 By the definition of~$G$ there is a rule $X_1 \btran{p} \beta$ with $p > 0$ and $\beta(X_2) \ge 1$.
 \begin{itemize}
  \item If $k > 2$, the successor configuration is still unstable, but with probability at least~$p$ its corresponding path in~$G$ has length at most $k-1$.
  \item If $k = 2$, we have with probability at least~$p$ that the successor configuration~$\alpha'$ satisfies $\alpha'(X_2) > \alpha(X_2) < K$.
 \end{itemize}
 Observe that if an increase $\alpha'(X_2) > \alpha(X_2) < K$ as described in the case $k=2$ happens,
  then the $X_2$-component will remain above $\alpha(X_2)$ as long as the cautious strategy is followed,
  because the cautious strategy will not schedule~$X_2$ as long as the $X_2$-component is at most~$K$.
 Moreover, such increases can happen only finitely often before all types are saturated.
 It follows that with probability~$1$ a stable configuration will be reached eventually.

 It is important to note that for cautious strategies the way how ties are resolved does not matter.
 Furthermore, for unstable $\alpha \in \N^\Gamma$ it does not matter if arbitrary types $X \in \Gamma$ with $\alpha(X) > K$ are scheduled in between.
 More precisely, for unstable~$\alpha$, consider two schedulers, say $\sigma_1$ and~$\sigma_2$,
  that both follow a cautious strategy but may schedule other types~$X$ with $\alpha(X) > K$ in between.
 If for each type $X \in \Gamma$ the same probabilistic outcomes occur in the same order when following $\sigma_1$ and~$\sigma_2$, respectively,
  then the resulting stable configurations $\beta_1$ and~$\beta_2$ satisfy $\down{\beta_1} = \down{\beta_2}$.
 In other words, differences can only occur in saturated types.

 Recall that in the finite-state MDP we need to redirect those transitions that lead to unstable configurations.
 We do that in the following way.
 For unstable~$\alpha$, let
  \[
   T_\alpha := \{\down{\beta} \mid \text{ $\beta$ is stable and reachable from $\alpha$ using a cautious strategy}\} \subseteq Q
  \]
 and let $p_\alpha : T_\alpha \to (0,1]$ be the corresponding probability distribution.
 As argued above, $p_\alpha$ does not depend on the particular choice of the cautious strategy.
 However, for the construction of the finite-state MDP one does not need to compute~$p_\alpha$,
  because for reachability with probability~$1$ in a finite-state MDP the exact values of nonzero probabilities do not matter.
 Note that we have $p_\alpha(q) > 0$ for all $q \in T_\alpha$.
 So if a scheduling action in the finite-state MDP would lead, in the infinite-state MDP, to an unstable configuration~$\alpha$ with probability~$p_0$,
  then in the finite-state MDP we replace this transition by transitions to~$T_\alpha$, each with probability $p_0 / |T_\alpha|$.
 As argued above, this reflects a cautious strategy (which is optimal) of the scheduler in the original infinite-state MDP for states outside of~$Q$.

 This redirecting needs to be done for all unstable~$\alpha$ that are reachable from~$Q$ within one step.
 There are only finitely many such~$\alpha$.

 The overall decision procedure is thus as follows:
 \begin{enumerate}
  \item Construct the finite-state MDP with $Q$ as set of states as described.
  \item Check whether $\down{\alpha}$ is winning with probability~$1$, where the target set is $Q \cap F$.
 \end{enumerate}
 If $\down{\alpha}$ is winning with probability~$1$, then there is a deterministic and memoryless scheduler.
 This scheduler can then be extended for the infinite-state MDP by a cautious strategy,
  resulting in a deterministic and memoryless scheduler.

 It remains to show how $T_\alpha$ can be computed.
 We compute~$T_\alpha$ using the decidability of the reachability problem for Petri nets.
 We construct a Petri net from~$\S$ that simulates cautious behaviour of the scheduler in unstable configurations.
 The set of places of the Petri net is $P := \Gamma \cup \{S_X \mid X \in \Gamma\}$, where the $S_X$ are fresh symbols.
 The intention is that a configuration $\alpha \in \N^P$ with $\alpha(S_X) = 1$ indicates that $X$ is saturated.

 For the transitions of the Petri net we need some notation.
 For $\alpha, \beta \in \N^P$ we write $\alpha \btran{}_\bullet \beta$ to denote a transition
   whose input multiset is~$\alpha$ and whose output multiset is~$\beta$.
 For $X \in P$ and $i \in \N$ we write $X^i$ to denote $\alpha \in \N^P$ with $\alpha(X) = i$ and $\alpha(Y) = 0$ for $Y \ne X$.
 For $\alpha \in \N^\Gamma$ we also write $\alpha$ to denote
  $\alpha' \in \N^P$ with $\alpha'(X) = \alpha(X)$ for $X \in \Gamma$ and $\alpha'(X) = 0$ for $X \not\in \Gamma$.

 We include transitions as follows.
 For each $X \btran{} \beta$ we include $X^{K+1} \btran{}_\bullet X^K + \beta$.
 This makes sure that the transition $X \btran{} \beta$ ``inherited'' from~$\S$ is only used ``cautiously'',
  i.e., in the presence of more than~$K$ processes of type~$X$.
 For each strongly connected component $\{X_1, \ldots, X_k\} \subseteq \Gamma$ of the graph~$G$ from the beginning of the proof,
  we include a transition $X_1^K + \cdots + X_k^K + \gamma \btran{}_\bullet \{S_{X_1}, \ldots, S_{X_k}\} + \gamma$
  with $\gamma = S_{Y_1} \ldots S_{Y_\ell}$ for $\{Y_1, \ldots, Y_\ell\} = \left( \bigcup_{i=1}^k \suc(X_i) \right) \setminus \{X_1, \ldots, X_k\}$.
 This reflects the definition of ``saturated'':
  the types of a strongly connected component $\{X_1, \ldots, X_k\}$ are saturated in a configuration~$\alpha \in \N^\Gamma$
  if and only if $\alpha(X_i) \ge K$ holds for all $1 \le i \le k$ and all successors are saturated.
 We also include transitions that ``suck out'' superfluous processes from saturated types: $\{S_X, X\} \btran{}_\bullet \{S_X\}$ for all $X \in \Gamma$.

 \newcommand{\petri}[1]{\langle #1 \rangle}%
 For a configuration $q \in Q$ we define $\petri{q} \in \N^P$ as the multiset with
  $\petri{q}(X) = q(X)$ and $\petri{q}(S_X) = 0$ for $X \not\in \Sat(q)$, and
  $\petri{q}(X) = 0$ and $\petri{q}(S_X) = 1$ for $X \in \Sat(q)$ (and hence $q(X) = K$).
 By the construction of the Petri net we have for all unstable $\alpha \in \N^\Gamma$ and all $q \in Q$
  that a stable configuration $\beta \in \N^\Gamma$ with $\down{\beta} = q$ is reachable from~$\alpha$ in~$\S$ using a cautious strategy
   if and only if
  $\petri{q}$ is reachable from~$\alpha$ in the Petri net.
 It follows that $T_\alpha$ can be computed for all unstable~$\alpha$.
\qed
\end{proof}

\begin{qtheorem}{\ref{thm-mdp-universal}}
 \stmtthmmdpuniversal
\end{qtheorem}
\begin{proof}
We extend the state space from $\N^\Gamma$ to $\N^\Gamma \times \N^\Gamma$.
A configuration $(c,a) \in \N^\Gamma \times \N^\Gamma$ contains the multiset~$c \in \N^\Gamma$ of current processes (as before),
 and an \emph{age} vector $a \in \N^\Gamma$ indicating for each $X \in \Gamma$ how many steps ago $X$ was last scheduled.
We take $a(X) = 0$ for those $X$ with $c(X) = 0$.
An MDP $\D'(\S)$ can be defined on this extended state space in the straightforward way:
 in particular, in each step in which an $X \in \Gamma$ with $c(X) > 0$ is not scheduled,
  the age $a(X)$ is increased by~$1$.
We emphasize this extension of the state space does \emph{not} result from changing the pBPP~$\S$,
 but from changing the MDP induced by~$\S$.

An age $a(X) \ge k$ indicates that $X$ was not scheduled in the last $k$ steps.
So a run $(c_0,a_0) (c_1,a_1) \ldots$ is fair if and only if $a_i(X) < k$ holds for all $i \in \N$ and all $X \in \Gamma$.
Define $G := (F \times \N^\Gamma) \cup (\N^\Gamma \times \{k, k+1, \ldots\}^\Gamma)$.
There is a natural bijection between the $k$-fair runs in~$\D(\S)$ avoiding~$F$ and all runs in~$\D'(\S)$ avoiding~$G$.
So it suffices to decide whether there exists a scheduler~$\sigma$ for~$\D'(\S)$
 with $\P_\sigma((\alpha_0, (0, \ldots, 0)) \models \Diamond G) < 1$.

To decide this we consider a turn-based game between two players, Scheduler (``he'') and Probability (``she'').
As expected, in configuration~$\alpha \in \N^\Gamma \times \N^\Gamma$ player Scheduler selects a type $X \in \Gamma$ with $\alpha(X) \ge 1$
 and player Probability picks~$\beta$ with $X \btran{} \beta$, leading to a new configuration $T(\alpha, X \btran{} \beta)$,
  where $T(\alpha, X \btran{} \beta) \in \N^\Gamma \times \N^\Gamma$ denotes the configuration obtained from~$\alpha$ by applying $X \btran{} \beta$
  according to the transitions of~$\D'(\S)$.
Despite her name, player Probability is not bound to obey the probabilities in~$\S$; rather she can pick~$\beta$ with $X \btran{} \beta$ as she wants.
The goal of Scheduler is to avoid~$G$; the goal of Probability is to hit~$G$.

We show that in this game one can compute the winning region for Probability.
We define sets $W_0 \subseteq  W_1 \subseteq \ldots$ with $W_i \subseteq \N^\Gamma \times \N^\Gamma$ for all $i \in \N$:
define $W_0 := G$ and for all $i \in \N$ define
 \[
  W_{i+1} := W_i \cup \{\alpha \in \N^\Gamma \times \N^\Gamma \mid
   \forall X \in \Gamma \ \exists \beta: X \btran{} \beta \text{ and } T(\alpha, X \btran{} \beta) \in W_i\}\,.
 \]
For all $i \in \N$ we have that $W_i$ is the set of configurations where Probability can force a win in at most $i$ steps.
As $W_0 = G$ is upward-closed with respect to the componentwise ordering $\mathord{\preceq}$ on $\N^\Gamma \times \N^\Gamma$,
 all $W_i$ are upward-closed with respect to~$\mathord{\preceq}$.
Considering the minimal elements of each~$W_i$, it follows by Dickson's lemma that for some $i$ we have $W_i = W_{i+1}$ and hence $W_j = W_i$ for all $j \ge i$.
Then $W := W_i$ is the winning region for Probability.
One can compute the minimal elements of~$W$ by computing the minimal elements for each~$W_1, \ldots, W_i = W$.
\newcommand{\tG}{\widetilde{G}}%
Define $\tG := (\N^\Gamma \times \N^\Gamma) \setminus W$, the winning region for Scheduler, a downward-closed set.

Next we show that for $\alpha_0 \in \N^\Gamma \times \N^\Gamma$ we have that
 there exists a scheduler~$\sigma$ for~$\D'(\S)$ with $\P_\sigma(\alpha_0 \models \Diamond G) < 1$
 if and only if in~$\D'(\S)$ there is a path from $\alpha_0$ to~$\tG$ avoiding~$G$.
For the ``if'' direction, construct a scheduler~$\sigma$ that ``attempts'' this path.
Since the path is finite, with positive probability, say~$p$, it will be taken.
Once in~$\tG$, the scheduler~$\sigma$ can behave according to Scheduler's winning strategy in the two-player game and thus avoid~$G$ indefinitely.
Hence $\P_\sigma(\alpha_0 \models \Diamond G) \le 1-p$.
For the ``only if'' direction, suppose that $\tG$ cannot be reached before hitting~$G$.
Then regardless of the scheduler the play remains in the winning region~$W$ of Probability in the two-player game.
Recall that $W = W_i$ for some $i \in \N$,
\newcommand{\pmin}{p_\mathit{min}}%
 hence regardless of the scheduler with probability at least $\pmin^i > 0$ the set~$G$ will, at any time, be reached within the next $i$ steps,
  where $\pmin > 0$ is the least positive probability occurring in the rules of~$\S$.
It follows that for all schedulers~$\sigma$ we have $\P_\sigma(\alpha_0 \models \Diamond G) = 1$.

It remains to show that it is decidable for a given $\alpha_0 \in \N^\Gamma \times \N^\Gamma$ whether there is path to~$\tG$ avoiding~$G$.
But this can be done using a Karp-Miller-style algorithm as in the proof of Theorem~\ref{thm-qual-reach-decid}.
In particular, the following analogue of Lemma~\ref{lem-qual-reach-decid} holds:
\begin{lemma} \label{lem-mdp-analogous}
 Let $\alpha_0 \in (\N^\Gamma \times \N^\Gamma) \setminus G$ and let $\gamma \in \N^\Gamma \times \N^\Gamma$.
 Let $\alpha_0 \to \alpha_1 \to \ldots \to \alpha_k$ be a shortest path in~$\D'(\S)$
  with $\alpha_0, \ldots, \alpha_k \not\in G$ and $\alpha_k \preceq \gamma$. 
 Then for all $i,j$ with $0 \le i < j \le k$ we have $\alpha_i \not\preceq \alpha_j$.
\end{lemma}
As argued in the proof of Theorem~\ref{thm-qual-reach-decid}
 one can build a (finite) tree of configurations reachable from~$\alpha_0$ via non-$G$-configurations and prune it whenever
  the path $\alpha_0 \to \ldots \to \alpha_k$ from the root to the current node~$\alpha_k$ contains a configuration $\alpha_j$ (where $0 \le j < k$)
   with $\alpha_j \preceq \alpha_k$.
If there is a path from~$\alpha_0$ to~$\tG$ via non-$G$-configurations, then the algorithms finds one.
\qed
\end{proof}

    \section{Proofs of Section~\ref{sec-semilinear}} \label{app-undecidability}

\begin{qtheorem}{\ref{thm-undecidability}}
\stmtthmundecidability
\end{qtheorem}
\begin{proof} \
\begin{itemize}
\item[(a)]
We will give a reduction from the complement
of the control-state reachability problem for (deterministic) 2-counter
machines (with counters $X$ and $Y$).
Given a 2-counter
machine $\CM = (\controls,\cmrules,q_0,q_F)$, we want to check if there is no
computation from configuration $(q_0,0,0)$ to any configuration in
$\{q_F\} \times \N^2$ in $\CM$. We will construct a pBPP $\S =
(\Gamma,\mathord{\btran{}},\Prob)$ that ``simulates'' $\CM$, a semilinear set $F
\subseteq \N^\Gamma$, and a configuration $\alpha_0 \in \N^\Gamma$ such that
$(q_0,0,0) \not\to^*_\CM \{q_F\} \times \N^2$ iff
$\P(\alpha_0 \models \Diamond F) = 1$. 

W.l.o.g., we assume that there is no transition from $q_F$ in $\cmrules$.
Define
\begin{eqnarray*}
    \Gamma & = & \controls \cup \{q_{bad}\} \cup \{X_+,X_-,Y_+,Y_-\}
            \cup \\
            &  & \{ (\theta_1,\theta_2) \mid \exists q,q',c_1,c_2 :
            \langle (q,\theta_1,\theta_2),(q',c_1,c_2)\rangle
            \in \cmrules \},
\end{eqnarray*}
where
$q_{bad},X_+,X_-,Y_+,Y_- \notin \controls$ and $\theta_1 \in 
\{X=0,X>0\}$ (resp. $\theta_2 \in \{Y=0,Y>0\}$)
is a zero test for first (resp. second) counters. Intuitively, in any pBPP
configuration
$\alpha$, the number $\alpha(X_+) - \alpha(X_-)$ (resp. $\alpha(Y_+) -
\alpha(Y_-)$) denotes the value of the first (resp. second) counter. For each
$n \in \Z$, let $\Sign(n) = +$ if $n \geq 0$; otherwise, let $\Sign(n) = -$.
For
each transition rule $(q,\theta_1,\theta_2) \to (q',c_1,c_2)$ in $\cmrules$,
we add to $\S$ a transition $q \btran{} q' (\theta_1,\theta_2)
X^{c_1}_{\Sign(c_1)} Y^{c_2}_{\Sign(c_2)}$, and a transition
$(\theta_1,\theta_2) \btran{} \varepsilon$. For each $Z \in (\controls
\setminus \{q_F\}) \cup \{X_+,X_-,Y_+,Y_-\}$, we also add $Z \btran{} q_{bad}$.
Finally, we add $q_F \btran{} q_F$. For each
transition in $\mathord{\btran{}}$, we do not actually care about its
actual probability; we simply set it to be strictly positive.

We now define $F$ to be the union of the following Presburger formulas:
\begin{enumerate}
\item $\neg\bigwedge_{(\theta_1,\theta_2) \in \Gamma}
    \left((\theta_1,\theta_2) > 0
    \to (\theta_1[(X_+ - X_-)/X] \wedge \theta_2[(Y_+ - Y_-)/Y])\right)$
\item $\neg\bigwedge_{(\theta_1,\theta_2) \in \Gamma} \left((\theta_1,\theta_2) > 0
    \to \left( (\theta_1,\theta_2) = 1 \wedge \bigwedge_{(\theta_1',\theta_2')
    \in \Gamma \setminus \{(\theta_1,\theta_2)\}} (\theta_1',\theta_2') = 0
    \right)\right)$
\item $q_{bad} > 0 \wedge q_F = 0$
\end{enumerate}
where $\theta_i[Z/Z']$ denotes replacing every occurrence of variable
$Z'$ in $\theta_i$ with the term $Z$.
The first conjunct above encodes bad configurations in $\S$ that are visited
if the simulation of $\CM$ by $\S$ is not faithful (i.e., counter tests are
violated but the transitions are still executed). Since $F$ is expressible
in Presburger Arithmetic, it follows that it is a semilinear set.

We now prove the correctness of the reduction, i.e., that
$(q_0,0,0) \not\to^*_\CM \{q_F\} \times \N^2$ iff
$\P(\alpha_0 \models \Diamond F) = 1$. Suppose that $(q_0,0,0) \to^*_\CM
(q_F,n_1,n_2)$ is witnessed by some path $\pi$, for some $n_1,n_2 \in \N$. This
implies that, for some $c_1,c_1',
c_2,c_2' \in \N$ such that $n_1 = c_1 - c_1'$ and $n_2 = c_2 - c_2'$, there is a
finite path $\pi'$ from $q_0$ to $\alpha = q_F X_+^{c_1} X_-^{c_1'} Y_+^{c_2}
Y_-^{c_2'}$ that avoids $F$.

This path is a faithful simulation of $\pi$,
which removes each $(\theta_1,\theta_2)$ as soon as it is introduced.
Since $\alpha \in \widetilde{F}$, it follows that $\P(\Run(\pi')) > 0$ and
so $\P( q_0  \models \Diamond F) < 1$. Conversely,
assume that $\P( q_0  \models \Diamond F) < 1$. It is easy to see that, for
the pBPP $\S$ that we defined above, both $\Mt(\S)$ and $\Mp(\S)$ when
restricted to states that are reachable from $\{q_0\}$ are globally
coarse. Thus, there exists
a finite path $\pi$ from $q_0$ to $q_F X_+^{c_1} X_-^{c_1'} Y_+^{c_2}
Y_-^{c_2'}q_{bad}^c$ that avoids $F$ (since each configuration in
$\widetilde{F}$ is of this form). In particular: (1) each time a symbol of the
form
$(\theta_1,\theta_2)$ is introduced in $\pi$, it is immediately removed, and
(2) rules of the form $Z \to q_{bad}$ (where $Z \in \{X_+,X_-,Y_+,Y_-\}$) are
executed only after $q_F$ is reached.
Therefore,
it follows that the path $\pi$ is a faithful simulation of $\CM$ and so corresponds
to a path $\pi': (q_0,0,0) \to^*_\CM (q_F,n_1,n_2)$ for some $n_1$ and $n_2$.
This completes the proof of correctness of the reduction.
\item[(b)]
The proof is a straightforward adaptation of the previous proof: for all $Z
\in \{X_+,X_-,Y_+,Y_-\}$, we add the rule $Z \btran{} Z$, and remove $Z
\btran{} q_{bad}$. That it suffices to restrict to 7-fair schedulers is because
there are a total of 6 types for $\S$.
\item[(c)]
The reduction is again from the acceptance problem of 2-counter machines.
As before, we are given a 2-counter machine $\CM =
(\controls,\cmrules,q_0,q_F)$ and we want to check if there is no
computation from configuration $(q_0,0,0)$ to any configuration in
$\{q_F\} \times \N^2$ in $\CM$. Without loss of generality, we may assume
that $q_0 \neq q_F$ and that there is no transition in $\cmrules$ of the form
$\langle q,(\theta_1,\theta_2),(q,c_1,c_2)\rangle$ (i.e. stay in the same
control state).  We now define the pBPP $\S$.
The set $\Gamma$ of process types is defined as follows:
\newcommand{\ON}{\bullet}
\newcommand{\OFF}{\circ}
\begin{eqnarray*}
    \Gamma & = & (\controls \cup \cmrules) \times \{\ON,\OFF\} \cup \\
           &   & \{X_+,X_-,Y_+,Y_-,Z,V\} \\
\end{eqnarray*}
Let $\bar \ON = \OFF$ and $\bar \OFF = \ON$. Let $J = \{X_+,X_-,Y_+,Y_-,Z,V\}$.
The initial configuration is $\alpha_0 = \{(q_0,\ON)\} \cup
\{ (q,\OFF): q \in \controls \setminus \{q_0\} \} \cup (\cmrules \times
\{\OFF\})$. The rules are as follows:
\begin{itemize}
\item for each $r \in \controls$ and $I \in \{\ON,\OFF\}$, we
have $(r,I) \btran{} (r,\bar I)Z$ and $(r,I) \btran{} (r,\bar I)ZV$.
\item for each $t \in \cmrules$, we have
$(t,\ON) \btran{} (t,\OFF)Z$ and $(t,\ON) \btran{} (t,\OFF)ZV$.
\item for each $t = \langle q,(\theta_1,\theta_2),(q',c_1,c_2) \rangle$,
    we have
    $(t,\OFF) \btran{} (t,\ON) X^{c_1}_{\Sign(c_1)}Y^{c_2}_{\Sign(c_2)}Z$ and
    $(t,\OFF) \btran{} (t,\ON) X^{c_1}_{\Sign(c_1)}Y^{c_2}_{\Sign(c_2)}ZV$
\item for each $W \in J$, we have $W \btran{} W$.
\end{itemize}
The semilinear target set $F$ is defined as a conjunction of the following
Presburger formulas:
\begin{itemize}
\item $(q_F,\ON) > 1$,
\item $V < 2$
\item If $V = 1$, then all of the following hold:
    \begin{itemize}
    \item $X_+ - X_- \geq 0$ and $Y_+ - Y_- \geq 0$
    \item for each $r \in \controls \cup \cmrules$, $(q,\ON) + (q,\OFF) = 1$
    \item If $Z \equiv 0 \pmod{4}$, then (a) for precisely one $q \in \controls$
    we have $(q,\ON) = 1$ and $(q',\OFF) = 1$ for each $q \neq q' \in
    \controls$, and
    (b) for each $t \in \cmrules$, we have $(t,\OFF) = 1$.
    \item If $Z \equiv 1 \pmod{4}$, then (a) for precisely one $q \in \controls$
    we have $(q,\ON) = 1$ and $(q',\OFF) = 1$ for each $q \neq q' \in
    \controls$,
    (b) for precisely one $t = \langle q,(\theta_1,\theta_2),q',(c_1,c_2)
    \rangle \in \cmrules$, we have $(t,\ON) = 1$ and for each $t \neq t'
    \in \cmrules$, we have $(t,\OFF) = 1$, and
    (c) $\theta_1[(X_+ - X_- - c_1)/X]$ and $\theta_2[(Y_+ - Y_- - c_2)/Y]$
    hold.
    \item If $Z \equiv 2 \pmod{4}$, then (a) for all $q \in \controls$
    we have $(q,\OFF) = 1$,
    (b) for precisely one $t = \langle q,(\theta_1,\theta_2),q',(c_1,c_2)
    \rangle \in \cmrules$, we have $(t,\ON) = 1$ and for each $t \neq t'
    \in \cmrules$, we have $(t,\OFF) = 1$, and
    (c) $\theta_1[(X_+ - X_- - c_1)/X]$ and $\theta_2[(Y_+ - Y_- - c_2)/Y]$
    hold.
    \item If $Z \equiv 3 \pmod{4}$, then (a) for precisely one $q \in \controls$
    we have $(q,\ON) = 1$ and $(q',\OFF) = 1$ for each $q \neq q' \in
    \controls$,
    (b) for precisely one $t = \langle q',(\theta_1,\theta_2),q,(c_1,c_2)
    \rangle \in \cmrules$, we have $(t,\ON) = 1$ and for each $t \neq t'
    \in \cmrules$, we have $(t,\OFF) = 1$, and
    (c) $\theta_1[(X_+ - X_- - c_1)/X]$ and $\theta_2[(Y_+ - Y_- - c_2)/Y]$
    hold.
    \end{itemize}
\end{itemize}
The target set $F$ above forces the scheduler to do a faithful simulation of
the input counter machine. For example, $(q_0,0,0) \to (q_3,1,0)$ via
the transition rule $t = (q_0,(X = 0,Y = 0),q_3,(1,0))$ is simulated by several
steps as follows (we omit mention of $(r,\OFF)$, when $(r,\OFF) > 0$):
\[
    (q_0,\ON) \btran{} (q_0,\ON) (t,\ON) Z X_+ \btran{} (t,\ON) Z^2 X_+
        \btran{} (q_3,\ON) (t,\ON) Z^3 X_+ \btran{} (q_3,\ON) Z^4 X_+
\]
As soon as the scheduler deviates from faithful simulation, the Probability
player can choose the rule that spawns $V$, which does not take us to $F$.
For the next non-looping move (i.e. not of the form $W \btran{} W$, which
does not help the scheduler), Probability can spawn another $V$ which
takes us to $\widetilde{F}$, which prevents the scheduler from reaching $F$
forever. Conversely, Probability cannot choose to spawn $V$ if Scheduler
performs a correct simulation; for, otherwise, $F$ will be reached in
one step. Therefore, this shows that there exists a scheduler $\sigma$ such that
$\P_\sigma(\alpha_0 \models \Diamond F) = 1$ iff the counter machine can reach
$q_F$ from $(q,0,0)$. \qed
\end{itemize}
\end{proof}

}{}

\end{document}